\documentclass[a4paper,envcountsame]{llncs}

\usepackage[utf8]{inputenc}
\usepackage[english]{babel}
\usepackage{amsmath}
\usepackage{stmaryrd}
\usepackage{amsfonts}
\usepackage{amssymb}
\usepackage{lmodern}
\usepackage{tikz-cd}
\usepackage[all]{xy}
\usepackage{algorithm}
\usepackage{algpseudocode}
\usepackage{float}
\usepackage{tikz}
\usepackage{wrapfig}
\usepackage[inline]{enumitem}
\usetikzlibrary{positioning,shapes,arrows,automata}
\tikzstyle{state}=[draw,ellipse]

\newcommand{\Sub}{\mathsf{Sub}}
\newcommand{\Pow}{\mathcal{P}}
\newcommand{\Cat}[1]{\mathcal{#1}}
\newcommand{\Set}{\mathsf{Set}}
\newcommand{\BA}{\mathsf{BA}}

\newcommand{\lfp}{\mathsf{lfp}}

\newcommand{\B}{B}
\newcommand{\id}{\mathsf{id}}

\newcommand{\theory}[2]{\mathit{th}_{#1}^{#2}}
\newcommand{\close}{\mathsf{close}}
\newcommand{\D}{\delta}
\newcommand{\xvert}[2]{\vert #1 \vert_{#2}}

\newcommand{\inl}{\mathit{inl}}
\newcommand{\inr}{\mathit{inr}}
\newcommand{\inlv}{\inl_{\vee}}
\newcommand{\inrv}{\inr_{\vee}}

\newcommand{\Sem}[1]{\llbracket #1 \rrbracket}
\newcommand{\diam}[1]{\left\langle #1 \right\rangle}

\newcommand{\counter}{\mathsf{counter}}

\newcommand{\Pred}{P}
\newcommand{\Q}{Q}

\newcommand{\op}{\mathsf{op}}

\newtheorem{assumption}[theorem]{Assumption}

\title{Coalgebra Learning via Duality}
\titlerunning{Learning via Duality}
\author{Simone Barlocco\inst{1}, Clemens Kupke\inst{1}\fnmsep\thanks{Partially supported by EPSRC grant EP/N015843/1.} \and Jurriaan Rot\inst{2}}
\authorrunning{S.~Barlocco, C.~Kupke and J.~Rot}
\institute{University of Strathclyde, Glasgow, \\
\email{\{simone.barlocco,clemens.kupke\}@strath.ac.uk}
\and
Radboud University, Nijmegen, \\
\email{j.rot@cs.ru.nl}}

\begin{document}
\maketitle

\begin{abstract}
Automata learning is a popular technique for inferring minimal automata through membership and equivalence queries. In this paper, we generalise learning to the theory of coalgebras. The approach relies on the use of logical formulas as tests, based on a dual adjunction between states and logical theories. This allows us to learn, e.g., labelled transition systems, using Hennessy-Milner logic. Our main contribution is an abstract learning algorithm, together with a proof of correctness and termination. 
\end{abstract}

\section{Introduction}
\label{sec:intro}

In recent years, automata learning is applied with considerable success 
to infer models of systems and in order to analyse and verify them.
%
Most current approaches to active automata learning are ultimately based 
on the original algorithm due to Angluin~\cite{Angluin87},
although numerous improvements have been made,
in practical performance and in extending the techniques
to different models~\cite{Vaandrager17}.


Our aim is to move from automata to \emph{coalgebras}~\cite{Rutten00,jacobs-coalg}, 
providing a generalisation of learning to a wide
range of state-based systems. 
The key insight underlying our work is that dual adjunctions connecting coalgebras and tailor-made logical languages~\cite{KupkeKP04,BonsangueK05,Klin07,PavlovicMW06,KupkeP11} allow us to devise a generic learning algorithm for coalgebras that is parametric in the type of system under consideration. 
Our approach gives rise to a fundamental distinction between \emph{states}
of the learned system and \emph{tests}, modelled as logical formulas. This distinction is blurred in the classical DFA algorithm, 
where tests are also used to specify the (reachable) states. It is precisely the distinction between
tests and states which allows us to
move beyond classical automata, and use, for instance, Hennessy-Milner logic to
learn bisimilarity quotients of labelled transition systems. 

%

To present learning via duality we need to introduce new notions and refine existing ones. 
First, in the setting of coalgebraic modal logic, we introduce the new notion of \emph{sub-formula closed}
collections of formulas, generalising suffix-closed sets of words in Angluin's algorithm (Section~\ref{sec:subform-closed}). 
Second, we import the abstract notion of \emph{base} of a functor from~\cite{alwin}, which allows
us to speak about `successor states' (Section~\ref{sec:base}). In particular, the base allows us to
characterise \emph{reachability} of coalgebras in a clear and concise way. This yields
a canonical procedure for computing the reachable part from a given initial state in a coalgebra,
thus generalising the notion of a generated subframe from modal logic.

We then rephrase \emph{coalgebra learning} as the problem of inferring a coalgebra which
is reachable, minimal and which cannot be distinguished from the original coalgebra 
held by the teacher using tests. This requires suitably adapting the computation of the reachable part
to incorporate tests, and only learn `up to logical equivalence'. 
We formulate the notion of \emph{closed table}, and an associated procedure to close tables. 
With all these notions in place, we can finally define our abstract algorithm for coalgebra learning,
together with a proof of correctness and termination (Section~\ref{sec:learning}). 
Overall, we consider this correctness and termination proof as the main contribution of the paper;
other contributions are the computation of reachability via the base and the
notion of sub-formula closedness. At a more conceptual level, our paper
shows how states and tests interact in automata learning, by rephrasing it in the context of a dual adjunction 
connecting coalgebra (systems) and algebra (logical theories). 
As such, we provide a new foundation of learning state-based systems.

\paragraph{Related work.}
%
The idea that tests in the learning algorithm should be formulas of a distinct logical language 
was proposed first in~\cite{baku18:angl}. However, the
work in {\em loc.cit.} is quite ad-hoc, confined to Boolean-valued modal logics,
and did not explicitly use duality. 
This paper is a significant improvement: the dual adjunction
framework and the definition
of the base~\cite{alwin} enables us to present a description of Angluin's algorithm in purely categorical terms, including
a proof of correctness and, crucially, termination. 
Our abstract notion of logic also enables us to recover {\em exactly} the standard DFA algorithm
(where tests are words) and the algorithm for learning Mealy machines (where test are many-valued), something that is not possible in \cite{baku18:angl}
where tests are modal formulas. 
Closely related to our work is also the line of research initiated by~\cite{jasi:auto14} and followed up within the CALF project~\cite{heer:anab16,HeerdtS017,heer17:lear} 
which applies ideas from category theory to automata learning. 
Our approach is orthogonal to CALF: the latter
focuses on learning a general version of {\em automata},
whereas our work is geared towards learning bisimilarity quotients of state-based transition systems.
While CALF lends itself to studying automata in a large variety
of base categories, our work thus far is concerned with varying the type of transition structures.

\section{Learning by Example}
\label{sec:examples}

The aim of this section is twofold: (i) to remind the reader of the key elements of Angluin's L$^*$ algorithm~\cite{Angluin87} and
(ii) to motivate and outline our generalisation. 

In the classical L$^*$ algorithm, the learner tries to learn
a regular language $\mathcal{L}$ over some alphabet $A$ or, equivalently, a DFA $\mathcal{A}$ accepting that language.
Learning proceeds by asking queries to a teacher who has access to this automaton.
{\em Membership queries} allow the learner to test whether a given word is in the language, 
and {\em equivalence queries} to test whether the correct DFA has been learned already.
The algorithm constructs so-called tables $(S,E)$ where $S, E \subseteq A^*$
are the rows and columns of the table, respectively. The value at position $(s,e)$ of the table is the answer to the membership query ``$se \in \mathcal{L}$?''.

Words play a double role: On the one hand, a word $w \in S$ represents the
state which is reached when reading $w$ at the initial state.
On the other hand, the set $E$ represents the set of membership queries that the learner is asking about the states in $S$.
A table is {\em closed} if for all $w \in S$ and all $a \in A$ either $wa \in S$ 
or there is a state $v \in S$ such that $wa$ is equivalent to $v$ w.r.t.\ membership queries of words in $E$.   
If a table is not closed we extend $S$ by adding words of the form $w a$ for $w \in S$ and $a \in A$.
Once it is closed, one can define a {\em conjecture},\footnote{The algorithm additionally
requires \emph{consistency}, but this is not needed if counterexamples
are added to $E$. This  idea goes back to~\cite{mapn95:onth}.} 
i.e., a DFA with states in $S$. The learner now asks
the teacher whether the conjecture is correct. If it is, the algorithm terminates. Otherwise 
the teacher provides a {\em counterexample}: a word on which the conjecture is incorrect. The table 
is now extended using the counterexample. As a result, the table
is not closed anymore and the algorithm continues again by closing the table.

Our version of L$^*$ introduces some key conceptual differences: tables are 
pairs 
$(S,\Psi)$ such that $S$ (set of rows) is a selection of states of $\mathcal{A}$ and $\Psi$ (set of columns) is a collection of tests/formulas.
Membership queries become checks of tests in $\Psi$ at states in $S$ and equivalence queries
verify whether or not the learned structure is logically equivalent to the original one. 
A table $(S,\Psi)$ is closed if for all successors $x'$ of elements of $S$ there exists an $x \in S$ such that $x$ and $x'$
are equivalent w.r.t.\ formulas in $\Psi$.
The clear distinction
between states and tests in our algorithm means that counterexamples are formulas that have to be added to $\Psi$. 
Crucially, the move from words to formulas 
allows us to use the rich theory of coalgebra and coalgebraic logic to devise a generic algorithm. 


We consider two examples within our generic framework: classical DFAs, yielding essentially the L$^*$ algorithm, 
and labelled transition systems, 
which is to the best of our knowledge not covered by standard automata learning algorithms.

For the DFA case, let $L = \{u \in \{a,b\}^{*} \mid \mbox{number of } a\mbox{'s} \mbox{ mod } 3 = 0\}$ 
and assume that the teacher uses the following (infinite) automaton describing $L$: 
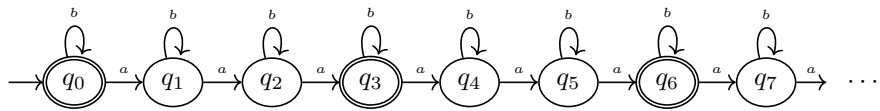
\begin{center}
\begin{tikzpicture}[->,node distance=1.3cm, semithick, auto]
  \tikzstyle{every state}=[text=black]

  \node[initial,state, initial text=, accepting] (A)  {$q_0$};
  \node[state] (B) [right of=A] {$q_1$};
  \node[state] (C) [right of = B]{$q_2$};
  \node[state] (D) [right of = C, accepting] {$q_3$};
  \node[state] (E) [right of = D] {$q_4$};
  \node[state] (F) [right of = E] {$q_5$};
  \node[state] (G) [right of = F, accepting] {$q_6$};
  \node[state] (H) [right of = G] {$q_7$};
  \node[state] (I) [right of = H, draw=none]{$\cdots$}; 
  
  \path (A) edge node {\tiny $a$} (B)
              edge [loop above] node {\tiny $b$} (C)
           (B) edge node {\tiny $a$} (C)
                       edge [loop above] node {\tiny $b$} (C)
           (C) edge node {\tiny $a$} (D)
                       edge [loop above] node {\tiny $b$} (C)
           (D) edge node {\tiny $a$} (E)
                       edge [loop above] node {\tiny $b$} (C)
           (E) edge node {\tiny $a$} (F)
                       edge [loop above] node {\tiny $b$} (C)
           (F) edge node {\tiny $a$} (G)
                       edge [loop above] node {\tiny $b$} (C)
           (G) edge node {\tiny $a$} (H)
                       edge [loop above] node {\tiny $b$} (C)
           (H) edge node {\tiny $a$} (I)
		         edge [loop above] node {\tiny $b$} (H);
\end{tikzpicture}
\end{center}
As outlined above, the learner starts to construct tables $(S,\Psi)$ where $S$ is a selection of states
of the automaton and $\Psi$ are formulas. For DFAs we will see (Ex.~\ref{ex:dfa}) that 
our formulas are just words in $\{a,b\}^*$.
Our starting table is $(\{q_0\}, \emptyset)$, i.e., we select the initial state and do not check any logical properties.
This table is trivially closed, as all states are equivalent w.r.t.\ $\emptyset$.
The first conjecture is the automaton consisting of one accepting state $q_0$ with $a$- and $b$-loops,
whose language is $\{a,b\}^*$.
This is incorrect and the teacher provides, e.g., $aa$ as counterexample.
The resulting table is $(\{q_0\},\{\varepsilon,a,aa\})$ where the second component was generated by closing $\{aa\}$
under suffixes. Suffix closedness features both in the original  L$^*$ algorithm and in our framework (Section~\ref{sec:subform-closed}). 
The table $(\{q_0\},\{\varepsilon,a,aa\})$ is not closed as $q_1$, the $a$-successor of $q_0$, 
does not accept $\varepsilon$ whereas $q_0$ does. Therefore we extend the table to $(\{q_0,q_1\},\{\varepsilon,a,aa\})$.
Note that, unlike in the classical setting, exploring
successors of already selected states cannot be achieved by appending letters to words, but we need to
{\em locally} employ the transition structure on the automaton $\mathcal{A}$ instead.
A similar argument shows that we need to extend the table further to $(\{q_0,q_1,q_2\},\{\varepsilon,a,aa\})$ which is closed.
This leads to the (correct) conjecture depicted on the right below.
The acceptance condition and transition structure has been
read off from the original automaton, where the transition from $q_2$ to $q_0$ is obtained by realising that $q_2$'s successor $q_3$
is represented by the equivalent state $q_0 \in S$. 

\begin{wrapfigure}[5]{r}{0cm}
\begin{minipage}{4cm}\vspace{-25pt}\centering
\begin{tikzpicture}[->,auto,node distance=1.2cm]]
  \tikzstyle{every state}=[text=black]

  \node[initial, initial text=, state, accepting] (A) {$q_0$};
  \node[state]         (B) [right of=A] {$q_1$};
  \node[state]         (C) [right of =B] {$q_2$}; 

  \path (A) edge [bend left]  node {\tiny $a$} (B)
            edge [loop above] node {\tiny $b$} (C)
        (B) edge [loop above] node {\tiny $b$} (C)
            edge [bend left]  node {\tiny $a$} (C)
        (C) edge [bend left]  node {\tiny $a$} (A)
            edge [loop above] node {\tiny $b$} (C);
\end{tikzpicture}
\end{minipage}
\end{wrapfigure}

 A key feature of our work is that the L$^*$ algorithm can be systematically generalised to new settings, in particular,
 to the learning of bisimulation quotients of transition systems.
Consider the following labelled transition system (LTS). 
We would like to learn its minimal representation, i.e., its quotient modulo bisimulation.

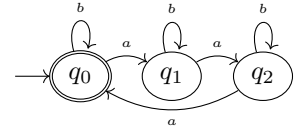
\begin{wrapfigure}[7]{r}{0cm}
\begin{minipage}{8cm}\vspace{-25pt}\centering
\begin{tikzpicture}[->,node distance=1.3cm, semithick, auto]
  \tikzstyle{every state}=[text=black]
  \node[state, initial, initial text=] (A)  {$x_0$};
  \node[state] (B) [below left of=A] {$x_1$};
  \node[state] (C) [below right of = A]{$x_2$};
  \node[state]
  (D) [below left of = B] {$x_3$};
  \node[state] 
  (E) [below right of = B] {$x_4$};
  \node[state] 
  (F) [right of = E] {$x_5$};
  \node[state] 
  (G) [right of = F] {$x_6$};
  \node[state] 
  (H) [right of = G] {$x_7$};
 \node[state] (I) [right of = H, draw=none]{$\cdots$}; 
  \path (A) edge [bend left = 20] node {\tiny $a$} (B)
                 edge [bend left = 20] node {\tiny $a$} (C)
           (B) edge [bend left = 20] node {\tiny $a$} (A)
                 edge node ['] {\tiny $b$} (D)
                 edge node {\tiny $a$} (E)
           (C) edge [bend left = 20] node {\tiny $a$} (A)
                edge node [', near start] {\tiny $b$} (F)
                edge node {\tiny $a$} (G)
           (E) edge [loop left] node {\tiny $b$} (E)
           (G) edge node {\tiny $b$} (H)
           (H) edge node {\tiny $b$} (I);
\end{tikzpicture}
\end{minipage}
\end{wrapfigure}
Our setting allows us to choose a suitable logical language.  
For LTSs, the language consists of the formulas of standard 
multi-modal logic (cf. Ex.~\ref{ex:Pow}).
The semantics is as usual where $\diam{a}\phi$ holds at a state if it has an $a$-successor that
makes $\phi$ true.

As above, the algorithm constructs tables, starting with $(S = \{x_0\}, \Psi = \emptyset)$.
The table is closed, so the first conjecture is a single state with an $a$-loop with no proposition letter true (note that $x_0$ has no $b$ or $c$ successor and 
no proposition is true at $x_0$). 
It is, however, easy for the teacher to find a counterexample. For example, the formula
$\diam{a} \diam{b} \top$ is true at the root of the original LTS but false in the conjecture.  
We add the counterexample and all its subformulas to $\Psi$ and obtain a new table
$(\{x_0\},\Psi'\}$ with $\Psi' = \{ \diam{a} \diam{b} \top, \diam{b} \top, \top\}$. 
Now, the table is not closed, as $x_0$ has successor $x_1$ that satisfies $\diam{b} \top$
whereas $x_0$ does not satisfy $\diam{b}\top$. 
Therefore we add $x_1$ to the table to obtain $(\{x_0,x_1\},\Psi')$.
Similar arguments will lead to the closed table
$(\{x_0,x_1,x_3,x_4\},\Psi')$ which also yields the correct conjecture.
Note that the state $x_2$ does not get added to the table as it is equivalent to $x_1$
and thus already represented. This demonstrates a remarkable fact: we computed the bisimulation
quotient of the LTS without inspecting the (infinite) right-hand side of the LTS.

Another important example that fits smoothly into our framework is the well-known variant of Angluin's algorithm to
learn Mealy machines~(Ex.~\ref{ex:mealy}). Thanks to our general notion of logic, our framework allows to 
use an intuitive language, where a formula is simply an input word $w$ whose truth value at a state $x$ is the observed output 
after entering $w$ at $x$. This is in contrast to~\cite{baku18:angl} where formulas had to be Boolean valued.
Multi-valued logics fit naturally in our setting; this is expected to be useful to deal with systems
with quantitative information.

\section{Preliminaries}
\label{sec:prelims}


The general learning algorithm in this paper is based on
the theory of \emph{coalgebras}, 
which provides an abstract framework for representing
state-based transition systems. 
In what follows we assume that the reader is familiar with basic notions of category theory 
and coalgebras~\cite{jacobs-coalg,Rutten00}.
We briefly recall the notion of pointed coalgebra, modelling
a coalgebra with an initial state. 
  Let $\Cat{C}$ be a category with a terminal object $1$ and let $\B \colon \Cat{C} \to \Cat{C}$ be a functor.
  A pointed $\B$-coalgebra is a triple $(X,\gamma,x_0)$ where $X \in \Cat{C}$ and  $\gamma\colon X \to \B X$
  and $x_0\colon 1 \to X$, specifying the coalgebra structure and the point (``initial state'') of the coalgebra, respectively.

%

\paragraph{Coalgebraic modal logic.}

Modal logics are used to describe properties of state-based systems, modelled here as coalgebras. 
The close relationship between coalgebras and their logics is described
elegantly via dual adjunctions~\cite{KupkeKP04,Klin07,PavlovicMW06,KupkeP11}.

Our basic setting consists of two categories $\Cat{C},
\Cat{D}$ connected by functors $\Pred, \Q$ forming a dual adjunction $P \dashv Q \colon \Cat{C} \leftrightarrows \Cat{D}^{\mathsf{op}}$.
In other words, we
have a natural bijection
$\Cat{C}(X, \Q \Delta) \cong \Cat{D}(\Delta, \Pred X) \mbox{ for } X \in 
\Cat{C}, \Delta \in \Cat{D}$.
Moreover, we assume 
\begin{wrapfigure}[4]{r}{0pt}
\begin{minipage}{16em} \vspace{-3em}
\begin{equation}\label{eq:dual-adjunction}
  \xymatrix@C=0.5cm{
\Cat{C}\ar@/^2ex/[rr]^-{\Pred} \save !L(.5) \ar@(dl,ul)^{B} \restore & \bot & \Cat{D}^{\mathsf{op}} \ar@/^2ex/[ll]^-{\Q} \save !R(.5) \ar@(ur,dr)^{L} \restore 
}
\end{equation}
\end{minipage}
\end{wrapfigure}
two functors, $\B \colon \Cat{C} \rightarrow \Cat{C}, L \colon
\Cat{D} \rightarrow \Cat{D}$, see~\eqref{eq:dual-adjunction}.
The functor $L$ represents the syntax of the (modalities in the) logic:
assuming that $L$ has an initial
algebra $\alpha \colon L \Phi \rightarrow \Phi$ we think of 
$\Phi$ as the collection of formulas, or tests.
In this logical perspective, the functor $\Pred$ maps
an object $X$ of $\Cat{C}$ to the collection of predicates
and the functor $\Q$ maps an object $\Delta$ of $\Cat{D}$
to the collection $\Q \Delta$ of $\Delta$-theories.


The connection between coalgebras and their logics is specified
via a natural transformation $\D \colon L \Pred \Rightarrow \Pred \B$, sometimes
referred to as the one-step semantics
\begin{wrapfigure}[5]{r}{0pt}
\begin{minipage}{18em} \vspace{-3.5em}
\begin{equation}\label{eq:semantics-logic}
\vcenter{
 \xymatrix@R=0.5cm{ L \Phi \ar@{-->}[r]^{L \Sem{\_}} \ar[d]_\alpha & L \Pred X \ar[r]^{\D_X} & 
     \Pred \B X \ar[d]^{\Pred \gamma} \\
    \Phi \ar@{-->}[rr]^{\exists ! \Sem{\_}} & & \Pred X }
}
\end{equation}
\end{minipage}
\end{wrapfigure}
 of the logic. The $\D$ is used to define the semantics of
the logic on a $\B$-coalgebra $(X,\gamma)$ by initiality,
as in~\eqref{eq:semantics-logic}.
Furthermore, using the bijective correspondence of the dual adjunction between $\Pred$ and $\Q$,
the map $\Sem{\_}$ corresponds to a map
$\theory{}{\gamma}\colon X \rightarrow \Q \Phi$ that we will refer to as the theory
map of $(X,\gamma)$. 
\begin{wrapfigure}[5]{r}{0pt}
\begin{minipage}{18em} \vspace{-3.5em}
\begin{equation}\label{eq:theory-map-logic}
\vcenter{
   \xymatrix@R=0.5cm{\B X \ar@{-->}[r]^{\B  \theory{}{\gamma}} & \B \Q \Phi \ar[r]^{\D^\flat_\Phi} 
      & \Q L \Phi \\  
      X \ar[u]^\gamma \ar@{-->}[rr]^{\exists ! \theory{}{\gamma}} & & \Q \Phi \ar[u]_{\Q \alpha}}
}
\end{equation}
\end{minipage}
\end{wrapfigure}

The theory map can be expressed directly via a universal property,
by making use of the so-called \emph{mate}  $\D^\flat\colon \B\Q \Rightarrow \Q L$
of the one-step semantics $\D$ (cf.~\cite{Klin07,PavlovicMW06}). 
More precisely, we have $\D^\flat = \Q L \varepsilon \circ \Q \delta \Q \circ \eta B \Q$,
where $\eta, \varepsilon$ are the unit and counit of the adjunction.
Then $\theory{}{\gamma}\colon X \to \Q \Phi$ is the unique morphism making~\eqref{eq:theory-map-logic} commute.
%

\begin{example}
\label{ex:dfa} Let $\Cat{C}= \Cat{D} = \mathsf{Set}, \Pred=\Q = 2^{-}$ the contravariant 
power set functor,
$\B = 2 \times -^{A}$ and $L = 1 + A \times
-$.  In this case $\B$-coalgebras can be thought of as 
deterministic automata with input alphabet $A$ (e.g.,~\cite{Rutten98}).
It is well-known that the initial $L$-algebra 
is $\Phi = A^{*}$ with structure $\alpha = [\varepsilon, \mathrm{cons}]\colon 1 + A \times A^* \to A^*$
where $\varepsilon$ selects the empty word and $\mathrm{cons}$ maps a pair $(a,w) \in A \times A^*$
to the word $aw \in A^*$, i.e., in this example our tests are words with the intuitive meaning that
a test succeeds if the word is accepted by the given automaton.
For $X \in \Cat{C}$, the $X$-component of the (one-step) semantics $\D\colon L \Pred \Rightarrow \Pred \B$ is defined as follows:
$ \D_X (\ast) = \lbrace (i, f) \in 2 \times X^A \mid i = 1 \rbrace$,
and $\D_X (a, U) = \lbrace (i, f) \in 2 \times X^A \mid f(a) \in U \rbrace$.
It is matter of routine checking that the semantics of tests in $\Phi$ 
on a $\B$-coalgebra $(X,\gamma)$ is as follows: we have
$\Sem{\varepsilon} = \{ x \in X \mid \pi_1 (\gamma(x)) = 1 \}$
and $\Sem{a w} = \{ x \in X \mid \pi_2 (\gamma(x))(a)\in \Sem{w}\}$,
where $\pi_1$ and $\pi_2$ are the projection maps.
The theory map $\theory{}{\gamma}$ sends a state to the language accepted by that state in the usual way. 
\end{example}
\begin{example}
\label{ex:mealy} Again let $\Cat{C} = \Cat{D} = \Set$ and consider the
functors $\Pred = \Q = O^{-}$, $\B = (O \times -)^A$ and $L = A \times (1 + -)$, where $A$ and $O$ are fixed sets,
thought of as input and output alphabet, respectively. Then $\B$-coalgebras
are Mealy machines and the initial $L$-algebra is given by the set $A^+$ of finite non-empty words over $A$.
For $X \in \Cat{C}$, the one-step semantics $\D_X\colon A \times (1 + O^X) \to O^{\B X}$ is defined by 
$\D_X(a,\mathrm{inl(*)}) = \lambda f . \pi_1 (f(a))$ and
$\D_X(a,\mathrm{inr}(g)) = \lambda f.  g(\pi_2(f(a)))$.
Concretely, formulas are words in $A^+$; the ($O$-valued) semantics of $w \in A^+$ at state $x$
is the output $o \in O$ that is produced after processing the input $w$ from state $x$.
\end{example}
\begin{example}\label{ex:Pow} 
Let $\Cat{C} = \Set$ and $\Cat{D} = \BA$, where the latter denotes the
category of Boolean algebras. 
Again $\Pred = 2^{-}$, but this time $2^X$ is interpreted as a Boolean algebra. 
The functor $\Q$ maps a Boolean algebra to
the collection of ultrafilters over it~\cite{blrive01:moda}.  Furthermore $\B=(\Pow -)^A$ where
$\Pow$ denotes covariant power set 
and $A$ a set of actions. 
Coalgebras for this functor correspond to
labelled transition systems, where a state has a set of successors that depends on the action/input from $A$.
The dual functor $L\colon \BA \to \BA$ is defined as
$L Y \mathrel{:=} F_{\BA} (\{ \diam{a} y \mid a \in A, y \in Y \}) 
/ \!\equiv$
where $F_\BA \colon \Set \to \BA$ denotes the free Boolean algebra functor and where, 
roughly speaking, $\equiv$ is the congruence generated from the axioms
    %
${\diam{a} \perp} \equiv {\perp}$ and $\diam{a}(y_1 \vee y_2) \mathrel{\equiv} \diam{a}(y_1) \vee \diam{a}(y_2)$ for each $a \in A$.
This is explained in more detail in~\cite{KupkeP11}. The initial algebra for this functor is the so-called Lindenbaum-Tarski
algebra~\cite{blrive01:moda} of modal formulas
$\left( \phi \mathrel{::=} \perp \mid \phi \vee \phi \mid \neg \phi \mid \diam{a} \phi \right)$ 
quotiented by logical equivalence. The definition of an appropriate $\D$
can be found in, e.g.,~\cite{KupkeP11}---the semantics $\Sem{\_}$ of a formula then amounts to the standard one~\cite{blrive01:moda}. 
\end{example}

Different types of probabilistic transition systems also fit into the dual adjunction framework, see, e.g,~\cite{JacobsS09}. 

\paragraph{Subobjects and intersection-preserving functors.}

We denote by $\Sub(X)$ the collection of subobjects of an object $X \in \Cat{C}$. 
Let $\leq$  be the order\label{order} on subobjects $s \colon S \rightarrowtail X, s' \colon S' \rightarrowtail X$ given by 
$s \leq s'$ iff there is $m \colon S \rightarrow S'$ s.t.\ $s = s' \circ m$. 
The \emph{intersection} $\bigwedge J \rightarrowtail X$ of a family $J = \{s_i \colon S_i \rightarrow X\}_{i \in I}$ is
defined as the greatest lower bound w.r.t.\ the order $\leq$. In a complete category, it
can be computed by (wide) pullback. We denote the maps in the limiting cone by $x_i \colon \bigwedge J \rightarrowtail S_i$. 

For a functor $B \colon \Cat{C} \rightarrow \Cat{D}$, we say $B$ \emph{preserves (wide)
intersections}
if it preserves these wide pullbacks, i.e., if $(B(\bigwedge J), \{B x_i\}_{i 
\in I})$ is the pullback of  $\{B s_i \colon B S_i \rightarrow B X\}_{i \in I}$.
By \cite[Lemma 3.53]{AMMS13} (building on~\cite{trnkova1971descriptive}), \emph{finitary} functors on $\Set$ `almost' preserve wide intersections: 
for every such functor $B$ there is a functor $B'$ which preserves wide intersections and agrees
with $B$ on all non-empty sets. 
Finally, if $B$ preserves intersections, then it preserves monos.

\paragraph{Minimality notions.}

The algorithm that we will describe in this paper learns a minimal and reachable representation of an object. 
The intuitive notions
of minimality and  reachability are formalised as follows. 

\begin{definition}
    We call a $\B$-coalgebra $(X,\gamma)$ \emph{minimal w.r.t.\ logical equivalence} if
    the theory map $\theory{}{\gamma}\colon X \to \Q \Phi$ is a monomorphism.
\end{definition}
\begin{definition}\label{def:reachable}
    We call a pointed $\B$-coalgebra $(X,\gamma,x_0)$ \emph{reachable}
    if for any subobject $s\colon S \to X$ and $s_0\colon 1 \to S$ with $x_0 = s \circ s_0$:
 	if $S$ is a subcoalgebra of $(X,\gamma)$ then
    $s$ is an isomorphism. 
\end{definition}
For expressive logics~\cite{schr08:expr}, behavioural equivalence concides with 
logical equivalence. Hence, in that case, our algorithm 
learns a ``well-pointed coalgebra'' in the terminology of~\cite{AMMS13}, i.e., a 
pointed coalgebra that is reachable and minimal w.r.t.~behavioural
equivalence. All logics appearing in this paper are expressive.


\paragraph{Assumption on $\Cat{C}$ and Factorisation System.}

Throughout the paper we will assume that $\Cat{C}$ is a complete
and well-powered category. Well-powered means that for each $X \in \Cat{C}$ the collection
$\Sub(X)$ of subobjects of a given object forms a set. 
Our assumptions imply~\cite[Proposition 4.4.3]{borceux1994}  that every morphism $f$ in $\Cat{C}$
 \begin{wrapfigure}[5]{r}{0pt}
\begin{minipage}{10em} \vspace{-3em}
\begin{equation}\label{eq:fill-in}
\begin{tikzcd}
X \ar[d, "h"'] \ar[r, "e", twoheadrightarrow] &
Y \ar[d, "g"] \ar[dl, "d", dashed] \\
U \ar[r, "m"', tail] & Z
\end{tikzcd}
\end{equation}
\end{minipage}
\end{wrapfigure}
factors uniquely (up to isomorphism) as $f = m \circ e$ with $m$
a mono and $e$ a strong epi. 
Recall that an epimorphism
$e \colon X \rightarrow Y$ is strong if for every commutative
square in~\eqref{eq:fill-in} where the bottom arrow is a monomorphism, there
exists a unique diagonal morphism $d$ such that the entire diagram commutes.


\section{Subformula Closed Collections of Formulas}
\label{sec:subform-closed}


%

Our learning algorithm will construct conjectures that are ``partially'' correct, i.e.,
correct with respect to a subobject of the collection of all formulas/tests. 
Recall this collection of all tests are formalised in our setting as the initial $L$-algebra $(\Phi, \alpha\colon L \Phi \to \Phi)$. 
To define a notion of partial correctness we need to consider 
subobjects of $\Phi$ to which we can restrict the theory map. This is formalised via the notion
of ``subformula closed'' subobject of $\Phi$. 

\begin{wrapfigure}[5]{r}{0pt}
\begin{minipage}{10em} \vspace{-3em}
\begin{equation}\label{eq:coalg-to-alg}
\begin{tikzcd}
LX \ar[r, "Lg^{\dagger}"] & LY \ar[d, "g"] \\
X \ar[u, "f"] \ar[r, "g^{\dagger}"] & Y 
\end{tikzcd}
\end{equation}
\end{minipage}
\end{wrapfigure}
The definition of such subobjects is based on the notion
of \emph{recursive coalgebra}. 
For $L \colon \Cat{D} \rightarrow \Cat{D}$ an endofunctor,
a coalgebra $f \colon X \rightarrow LX$ is called
\textit{recursive} if for every $L$-algebra $g \colon
LY \rightarrow Y$ there is a unique `coalgebra-to-algebra' map 
$g^{\dagger}$ making~\eqref{eq:coalg-to-alg} commute.




\begin{definition}\label{def:subclosed}
 A subobject $j \colon \Psi \to \Phi$ is called a {\em subformula closed collection} (of formulas) 
 if there is a unique
 $L$-coalgebra structure $\sigma \colon \Psi \to L \Psi$ such that
	$(\Psi, \sigma)$ is a recursive $L$-coalgebra and 
 $j$ is the (necessarily unique) coalgebra-to-algebra map
  from $(\Psi,\sigma)$ to the initial algebra $(\Phi,\alpha)$.
\end{definition}

 \begin{remark}
   The uniqueness of $\sigma$ in Definition~\ref{def:subclosed} is implied if $L$ preserves 
   monomorphisms. This is the case in our examples. The notion of recursive coalgebra 
   goes back to~\cite{tayl99:prac,OSIUS197479}. 
   The paper~\cite{AdamekLM07} contains a claim that
   the first item of our definition of subformula closed collection is implied by the second one if $L$ preserves preimages.   
   In our examples both properties of $(\Psi,\sigma)$ are verified directly, rather than by relying on general
   categorical results. 
 \end{remark}


\begin{example}
     In the setting of Example~\ref{ex:dfa}, where the initial
     $L$-algebra is based on the set $A^*$ of words over the set (of inputs) $A$, a subset $\Psi \subseteq A^*$ is subformula-closed
     if it is suffix-closed, i.e., if for all $a w \in \Psi$ we have $w \in \Psi$ as well. 
\end{example}
\begin{example}
	In the setting that $\B = (\Pow -)^A$ for some set of actions $A$, $\Cat{C} = \Set$ 
     and $\Cat{D}=\BA$, the logic is given as a functor $L$ on Boolean algebras as discussed in Example~\ref{ex:Pow}.
     As a subformula closed collection is an object in $\Psi$, we are not simply dealing with a set of formulas, but with a Boolean algebra.
     The connection to the standard notion of being closed under taking subformulas in modal logic~\cite{blrive01:moda} can be sketched as follows:
     given a set $\Delta$ of modal formulas that is closed under taking subformulas, we define a Boolean algebra $\Psi_\Delta \subseteq \Phi$
     as the smallest Boolean subalgebra of $\Phi$ that is generated by the set $\hat{\Delta} = \{ [\phi]_{\Phi} \mid \phi \in \Delta \}$ where for a formula
     $\phi$ we let $[\phi]_\Phi \in \Phi$ denote its equivalence class in $\Phi$. 
     
     It is then not difficult to define a suitable $\sigma\colon \Psi_\Delta \to L \Psi_\Delta$.
     As $\Psi_\Delta$ is generated by closing  $\hat{\Delta}$ under Boolean operations, any
     two states $x_1,x_2$ in a given coalgebra $(X,\gamma)$ satisfy 
     $ \left( \forall b \in \Psi_\Delta. x_1 \in \Sem{b}  \Leftrightarrow x_2 \in \Sem{b} \right) \mbox{ iff }
      \left( \forall b \in \hat{\Delta}. x_1 \in \Sem{b}  \Leftrightarrow x_2 \in \Sem{b} \right).
     $
     In other words, equivalence w.r.t.\ $\Psi_\Delta$ coincides with equivalence w.r.t.\ 
     the {\em set} of formulas $\Delta$. This explains why in the concrete algorithm, 
     we do not deal with Boolean algebras explicitly, but with
     subformula closed sets of formulas instead.     
\end{example}

\begin{wrapfigure}[5]{r}{0pt}
\begin{minipage}{18em} \vspace{-3.5em}
	\begin{equation}\label{eq:thmap-psi}
	\vcenter{
		\xymatrix@C=1cm{
	X \ar[d]_{\gamma} \ar[rr]^{\theory{\Psi}{\gamma}}
		& & \Q \Psi  \\
	\B X \ar[r]^{\B \theory{\Psi}{\gamma}} 
		& \B \Q \Psi \ar[r]^{\delta^{\flat}_{\Psi}}
		& \Q L\Psi \ar[u]_{\Q \sigma}
        } }
    \end{equation}
\end{minipage}
\end{wrapfigure}
The key property of subformula closed collections $\Psi$ is that we can restrict
our attention to the so-called $\Psi$-theory map. Intuitively, subformula closedness 
is what allows us to define this theory map inductively. 


\begin{lemma}\label{lm:subf}
	Let $\Psi \stackrel{j}{\rightarrowtail} \Phi$ be a sub-formula closed collection, with coalgebra structure $\sigma \colon \Psi \rightarrow L\Psi$. 
	Then $\theory{\Psi}{\gamma} = \Q j \circ \theory{\Phi}{\gamma}$ is the unique
	map making~\eqref{eq:thmap-psi} commute.
We call $\theory{\Psi}{\gamma}$ the $\Psi$-theory map, and omit the $\Psi$
if it is clear from the context.
\end{lemma}

%
%
%

\section{Reachability and the Base}
\label{sec:base}


In this section, we define the notion of \emph{base} of an endofunctor, taken from~\cite{alwin}.
This allows us to speak about the (direct) successors of states in a coalgebra,
and about reachability, which
are essential ingredients of the learning algorithm.


\begin{definition}
Let $\B\colon \Cat{C} \rightarrow \Cat{C}$ be an
endofunctor.
We say $\B$ \emph{has a base} if for every arrow $f
\colon X \rightarrow \B Y$ there exist $g \colon 
X \rightarrow \B Z$ and $m \colon Z \rightarrowtail Y$
with $m$ a monomorphism such that $f = \B m \circ g$, and for any
pair $g' \colon X \rightarrow \B Z', m' \colon Z' \rightarrowtail Y$
with $\B m' \circ g' = f$ and $m'$ a monomorphism there is a unique
arrow $h \colon Z \rightarrow Z'$ such that $\B h \circ g
= g'$ and $m' \circ h = m$, see Diagram~\eqref{eq:base-diagram}.
We call $(Z,g,m)$ the \emph{($\B$)-base} of the morphism $f$. 
\end{definition}

\begin{wrapfigure}[5]{r}{0pt}
\begin{minipage}{17em} \vspace{-3.5em}
\begin{equation}\label{eq:base-diagram}
\begin{tikzcd}
X \ar[rr, "f", bend left] \ar[r, "g"'] 
  \ar[dr, "g'"', bend right] 
  & \B Z  \ar[d, "\B h"] \ar[r, "\B m"'] & \B Y \\
  & \B Z' \ar[ur, "\B m'"', bend right] & 
\end{tikzcd}
\end{equation}
\end{minipage}
\end{wrapfigure}
We sometimes refer to $m \colon Z \rightarrowtail Y$ as the base of $f$, omitting the $g$
when it is irrelevant, or clear from the context. 
Note that the terminology `the' base is justified, as it is easily
seen to be unique up to isomorphism. 

For example, 
let $B \colon \Set \rightarrow \Set$, $BX = 2 \times X^A$.
The base of a map $f \colon X \rightarrow BY$ is given by 
$m \colon Z \rightarrowtail Y$, where 
$Z = \{(\pi_2 \circ f)(x)(a) \mid x \in X, a \in A \}$,
and $m$ is the inclusion. The associated $g \colon X \rightarrow BZ$
is the corestriction of $f$ to $BZ$. 

For $\B = (\Pow -)^A \colon \Set \rightarrow \Set$,
the $\B$-$base$ of $f \colon X \rightarrow Y$ is given by the inclusion $m \colon Z \rightarrowtail Y$,
where $Z=\{ y \in Y \mid \exists x \in X, \exists a \in A \mbox{ s.t. } y \in f(x)(a)\}$.

%
%
%
%
%

\begin{proposition}\label{prop:existence-base}
	Suppose $\Cat{C}$ is complete and well-powered, and $B \colon \Cat{C} \rightarrow
	\Cat{C}$ preserves (wide) intersections. Then $B$ has a base. 
\end{proposition}
If $\Cat{C}$ is a locally presentable category, 
then it is complete and well-powered~\cite[Remark 1.56]{AR94}. Hence, in that case,
any functor $B \colon \Cat{C} \rightarrow \Cat{C}$ which preserves intersections has a base.
The following lemma will be useful in proofs.

\begin{lemma}\label{lm:nat-base}
Let $\B \colon \Cat{C} \rightarrow \Cat{C}$ be a functor that has a base and that preserves pre-images. 
Let $f \colon S \rightarrow \B X$ and
$h \colon X \rightarrow Y$ be morphisms, let
$(Z,g,m)$
be the base of $f$
and let $e \colon Z \rightarrow W, m' \colon W \rightarrow
Y$ be the (strong epi, mono)-factorisation of $h \circ m$. 
Then $(W,Be \circ g, m')$ 
is the base of $\B h \circ f$.
\end{lemma}


The $\B$-base provides an elegant way to relate reachability within a coalgebra
to a monotone operator on the (complete) lattice of subobjects of the carrier of the coalgebra.
Moreover, we will see that the least subcoalgebra that contains a given subobject of the
carrier can be obtained via a standard least fixpoint construction. Finally, we will introduce
the notion of prefix closed subobject of a coalgebra, generalising the prefix closedness
condition from Angluin's algorithm. 


By our assumption on $\Cat{C}$ at the end of Section~\ref{sec:prelims}, 
the collection of subobjects $(\Sub(X),\leq)$ 
ordered as usual (cf.~page~\ref{order}) forms a complete lattice. 
%
%
Recall that the meet on $\Sub(X)$ (intersection) is defined via pullbacks.  
In categories
with coproducts, the join $s_1 \vee s_2$ of subobjects 
$s_1,s_2 \in \Sub(X)$ is defined as the mono part of the factorisation of the map
$[s_1,s_2]\colon S_1 + S_2 \to X$, i.e., $[s_1,s_2] = (s_1 \vee s_2) \circ e$ for a strong epi
$e$. In $\Set$, this amounts to taking the union of subsets. 

\begin{wrapfigure}[4]{r}{0pt}
\begin{minipage}{13.5em} \vspace{-2.5em}
\begin{equation}\label{eq:operator-diagram}
\vcenter{
   \xymatrix@R=0.8cm{S \ar[d]_g \ar[r]^{s} & X \ar[d]^\gamma \\
      \B \Gamma(S) \ar[r]^-{\B \Gamma_\gamma^\B(s)} & \B X }
}
\end{equation}
\end{minipage}
\end{wrapfigure}
For a binary join $s_1 \vee s_2$ we denote by $\inlv\colon S_1 \to (S_1 \vee S_2)$ and $\inrv\colon S_2  \to (S_1 \vee S_2)$
the embeddings that exist by $s_i \leq s_1 \vee s_2$ for $i = \{1,2\}$. 
Let us now define the key operator of this section.

\begin{definition}\label{def:gamma}
  Let $\B$ be a functor that has a base, $s \colon S \rightarrowtail X$ a subobject of some $X \in \Cat{C}$ 
  and let $(X,\gamma)$ be a $\B$-coalgebra. Let $(\Gamma(S), g, \Gamma_\gamma^\B(s))$ be the $B$-base of
  $\gamma \circ s$, see Diagram~\eqref{eq:operator-diagram}.
  Whenever $\B$ and $\gamma$ are clear from the context, we write
  $\Gamma(s)$ instead of $\Gamma_\gamma^\B(s)$.
\end{definition}

\begin{lemma}\label{lm:gamma-monotone}
  Let $\B \colon \Cat{C} \to \Cat{C}$ be a functor with a base and let $(X,\gamma)$ be a $\B$-coalgebra. The operator $\Gamma\colon \Sub(X) \to \Sub(X)$ defined by
  $s \mapsto \Gamma(s)$ is monotone.
\end{lemma}
Intuitively, $\Gamma$ computes for a given set of states $S$ the set of ``immediate successors'', i.e., 
the set of states that can be reached by applying $\gamma$ to an element of $S$.
We will see that pre-fixpoints of $\Gamma$ correspond to subcoalgebras. Furthermore,
$\Gamma$ is the key to formulate our notion of closed table in the learning algorithm.

\begin{proposition}\label{prop:subcoalg-base}
Let $s \colon S \rightarrowtail X$ be a subobject and $(X, \gamma) \in \mathsf{Coalg}(\B)$ for $X \in \Cat{C}$
and $\B\colon\Cat{C} \to \Cat{C}$ a functor that has a base. Then $s$ is a subcoalgebra
of $(X, \gamma)$ if and only if $\Gamma(s) \leq s$. Consequently, the collection of
subcoalgebras of a given $\B$-coalgebra forms a complete lattice. 
\end{proposition}
Using this connection, 
reachability of a pointed coalgebra (Definition~\ref{def:reachable}) can be expressed
in terms of the least fixpoint $\lfp$ of an operator defined in terms of $\Gamma$. 
\begin{theorem}\label{thm:reach-base}
	Let $\B\colon\Cat{C} \to \Cat{C}$ be a functor that has a base.
	A pointed $B$-coalgebra $(X,\gamma,x_0)$ is reachable 
	iff $X \cong \lfp(\Gamma \vee x_0)$ (isomorphic as subobjects of $X$, i.e., equal).
\end{theorem}
This justifies defining the reachable part from an initial state $x_0 \colon 1 \rightarrowtail X$ as the least fixpoint
of the monotone operator $\Gamma \vee x_0$. Standard means of computing the least fixpoint
by iterating this operator then give us a way to compute this subcoalgebra.
Further, $\Gamma$ provides a way to generalise the notion of ``prefixed closedness'' from Angluin's L$^*$ algorithm
to our categorical setting.
\begin{definition}
  Let $s_0,s \in \Sub(X)$ for some $X\in \Cat{C}$ and let $(X,\gamma)$ be a $\B$-coalgebra.
  We call $s$ {\em $s_0$-prefix closed w.r.t.\ $\gamma$} if 
  $s = \bigvee_{i=0}^n s_i$ for some $n \geq 0$ and a collection $\{s_i \mid i = 1,\ldots,n\}$ with 
  $s_{j + 1} \leq  \Gamma(\bigvee_{i=0}^j s_i)$ for all $j$ with $0 \leq j < n$.
\end{definition}

\section{Learning Algorithm}
\label{sec:learning}

We define a general learning algorithm 
for $B$-coalgebras. 
First, we describe the setting, in general and slightly informal terms.
The teacher has
a pointed $B$-coalgebra $(X, \gamma, s_0)$. 
Our task is to `learn' 
a pointed $B$-coalgebra $(S, \hat{\gamma}, \hat{s}_0)$ s.t.:
\begin{itemize}
	\item $(S,\hat{\gamma}, \hat{s}_0)$ is \emph{correct} w.r.t.\ the
	collection $\Phi$ of all tests, i.e., the theory of $(X,\gamma)$
	and $(S,\hat{\gamma})$ coincide on the initial states $s_0$ and $\hat{s}_0$, 
	(Definition~\ref{def:correct-psi});
	\item $(S,\hat{\gamma},\hat{s}_0)$ is minimal w.r.t.\ logical equivalence;
	\item $(S,\hat{\gamma},\hat{s}_0)$ is reachable. 
\end{itemize}
The first point means that the learned coalgebra is `correct', that is,
it agrees with the coalgebra of the teacher on all possible tests from the initial state.
For instance, in case of deterministic automata and their logic in Example~\ref{ex:dfa},
this just means that the language of the learned automaton is the correct one. 

In the learning game, we are only provided limited access to the coalgebra
$\gamma \colon X \rightarrow B{X}$. Concretely, the teacher gives us:
\begin{itemize}
	\item 
	for any subobject $S \rightarrowtail X$ and sub-formula closed subobject $\Psi$ of $\Phi$,
	the composite theory map 
	$
	\begin{tikzcd}
		S \ar[r, tail]
			& X \ar[r, "{\theory{\Psi}{\gamma}}"]
			& \Q \Psi 
	\end{tikzcd}
	$;
	\item for $(S, \hat{\gamma}, \hat{s}_0)$ a pointed coalgebra, whether
	or not it is correct w.r.t.\ the collection $\Phi$ of all tests;
	\item in case of a negative answer to the previous question,
	a \emph{counterexample}, which essentially
	is a subobject $\Psi'$ of $\Phi$ representing some tests on which
	the learned coalgebra is wrong (defined more precisely below);
	\item for a given subobject $S$ of $X$, the `next states'; formally, the computation of the $B$-base of
	the composite arrow 
	$
	\begin{tikzcd}
		S \ar[r,tail]
			& X \ar[r,"\gamma"]
			& B{X} 
	\end{tikzcd}
	$.
\end{itemize}
The first three points correspond respectively to the standard notions of membership query (`filling
in' the table with rows $S$ and columns $\Psi$),
equivalence query and counterexample generation. The last point, about the base,
is more unusual: it does not occur in the standard algorithm, since there a canonical
choice of $(X,\gamma)$ is used, which allows to represent next states in a fixed manner. 
It is required in our setting of an arbitrary coalgebra $(X,\gamma)$. 

In the remainder of this section, we describe the abstract learning algorithm 
and its correctness. First, we describe the basic ingredients needed for the algorithm: tables, closedness,
counterexamples and a procedure to close a given table (Section~\ref{sec:tables}). 
Based on these notions, the actual algorithm is presented (Section~\ref{sec:alg}),
followed by proofs of correctness and termination (Section~\ref{sec:correctness-and-termination}).

\begin{assumption}
	Throughout this section, we assume
	\begin{itemize}
		\item that we deal with coalgebras over the base category $\Cat{C} = \Set$;
		\item a functor $B \colon \Cat{C} \rightarrow \Cat{C}$ that preserves pre-images and wide intersections; 
		\item a category $\Cat{D}$ with an initial object $0$ s.t.\ arrows with domain $0$ are monic;
		\item a functor $L \colon \Cat{D} \rightarrow \Cat{D}$ 
		with an initial algebra $L \Phi \stackrel{\cong}{\rightarrow} \Phi$;
		\item an adjunction $\Pred \dashv \Q  \colon \Cat{C} \leftrightarrows \Cat{D}^\op$,
		and a logic $\D \colon L \Pred \Rightarrow \Pred \B$.
	\end{itemize}
	Moreover, we assume 
	a pointed $B$-coalgebra $(X, \gamma, s_0)$. 
\end{assumption}
\begin{remark}\label{rem:assumption-disc}
  We restrict to $\Cat{C} = \Set$, but see it as a key contribution
  to state the algorithm in categorical terms: the assumptions
  cover a wide class of functors on $\Set$, which is the main direction of generalisation.
  Further, the categorical approach will enable future generalisations. 
  The assumptions on the category $\Cat{C}$ are: it is complete,
  well-powered and satisfies that for all (strong) epis $q\colon S \to \overline{S} \in \Cat{C}$ 
  and all monos $i \colon S' \to S$ such that $q \circ i$ is mono there
  is a morphism $q^{-1} \colon \overline{S} \to S$ such that (i) $q \circ q^{-1} = \id$ and $q^{-1} \circ q \circ i = i$.
\end{remark}

\subsection{Tables and counterexamples}\label{sec:tables}

\begin{definition}
	A \emph{table} is a pair $(S \stackrel{s}{\rightarrowtail} X, \Psi \stackrel{i}{\rightarrowtail} \Phi)$
	consisting of a subobject $s$ of $X$ and a subformula-closed subobject $i$ of $\Phi$. 
\end{definition} 
To make the notation a bit lighter, we sometimes refer to a table by $(S,\Psi)$, using $s$
and $i$ respectively to refer to the actual subobjects. 
%
%
%
The pair $(S,\Psi)$ represents `rows' and `columns' respectively, in the table;
the `elements' of the table are given abstractly by the map $\theory{\Psi}{\gamma} \circ s$.
In particular, if $\Cat{C} = \Cat{D} = \Set$ and $\Q = 2^{-}$, then this is a map
$S \rightarrow 2^\Psi$, assigning a Boolean value to every pair of a row (state) and a column (formula).

\begin{wrapfigure}[4]{r}{0pt}
\begin{minipage}{18em} \vspace{-3.3em}
\begin{equation}\label{eq:closed-diagram}
	\begin{tikzcd}
		S \ar[r, "s", tail] &
		X \ar[r, "\theory{}{\gamma}"] &
		\Q \Psi  \\
		\Gamma(S) \ar[r, "{\Gamma(s)}"'] \ar[u, "k"]
			& X \ar[ur,"\theory{}{\gamma}"']
			&
	\end{tikzcd}
\end{equation}
\end{minipage}
\end{wrapfigure}
For the definition of closedness, we use the operator $\Gamma(S)$ from Definition~\ref{def:gamma},
which characterises the successors of a subobject $S \rightarrowtail X$. 
\begin{definition}\label{def:closed}
	A table $(S,\Psi)$ is \emph{closed} if there exists a map
	$k \colon \Gamma(S) \rightarrow S$ such that Diagram~\eqref{eq:closed-diagram} commutes. 
	A table $(S,\Psi)$ is \emph{sharp} if the composite map
	$
	\begin{tikzcd}
		S \ar[r, "s"] &
		X \ar[r, "{\theory{}{\gamma}}"] &
	 	\Q \Psi
	\end{tikzcd}
	$
	is monic.
\end{definition}
Thus, a table $(S,\Psi)$ is closed if all the successors of states (elements of $\Gamma(S)$) are already
represented in $S$, up to equivalence w.r.t.\ the tests in $\Psi$. In other terms,
the rows corresponding to successors of existing rows are already in the table. 
Sharpness amounts to minimality w.r.t.\ logical equivalence: every row has a unique value. 
The latter will be an invariant of the algorithm (Theorem~\ref{thm:invariant}).

\begin{wrapfigure}[4]{r}{0pt}
\begin{minipage}{17em} \vspace{-3.3em}
\begin{equation}\label{eq:conj}
	\begin{tikzcd}
	S \ar[r, "s", tail] \ar[d, "\hat{\gamma}"'] &
	X \ar[r, "\gamma"] &
	B X  \ar[d, "B\theory{}{\gamma}"] \\
	B S \ar[r, "Bs"'] &
	B X \ar[r, "B\theory{}{\gamma}"'] &
	B \Q \Psi
	\end{tikzcd}
	\end{equation}
\end{minipage}
\end{wrapfigure}
A \emph{conjecture} is a coalgebra on $S$, which is not quite a subcoalgebra of $X$: instead,
it is a subcoalgebra `up to equivalence w.r.t.\ $\Psi$', that is, the successors agree
up to logical equivalence. 
\begin{definition}\label{def:conjecture}
	Let $(S,\Psi)$ be a table. 
	A coalgebra structure $\hat{\gamma} \colon S \rightarrow B S$ is called a \emph{conjecture} (for $(S,\Psi)$) if
	Diagram~\eqref{eq:conj} commutes.
\end{definition}
It is essential to be able to construct a conjecture from a closed table. The following, stronger 
result is a variation of Proposition~\ref{prop:subcoalg-base}.
\begin{theorem}\label{thm:conjecture}
	A sharp table is closed iff there 
	exists a conjecture for it. 
	Moreover, if the table is sharp and $B$ preserves monos, then this conjecture is unique. 
\end{theorem}

\begin{wrapfigure}[5]{r}{0pt}
\begin{minipage}{17em} \vspace{-2.5em}
\begin{equation}\label{eq:correct-diagram}
\begin{tikzcd}
	& X \ar[dr, "{\theory{}{\gamma}}"] & \\
1 \ar[r, "{\hat{s}_0}"', tail] \ar[ur, "s_0", tail]  &
	S \ar[r, "{\theory{}{\hat{\gamma}}}"']
		& \Q \Psi
\end{tikzcd}
\end{equation}	
\end{minipage}
\end{wrapfigure}
\noindent Our goal is to learn a pointed coalgebra which is correct w.r.t.\ 
all formulas. To this aim we ensure correctness w.r.t.\ an in\-crea\-sing sequence
of subformula closed collections $\Psi$. 

\begin{definition}\label{def:correct-psi}
	Let $(S,\Psi)$ be a table, and let
	$(S,\hat{\gamma},\hat{s}_0)$ be a pointed $B$-coalgebra on $S$.
We say $(S,\hat{\gamma}, \hat{s}_0)$ is \emph{correct} w.r.t.\ $\Psi$ if Diagram~\eqref{eq:correct-diagram}
commutes.

\end{definition}
All conjectures constructed during the learning algorithm will be correct w.r.t.\ the
subformula closed collection $\Psi$ of formulas under consideration.

\begin{lemma}\label{lm:truth-lemma-tables}
	Suppose $(S,\Psi)$ is closed, and $\hat{\gamma}$ is a conjecture. 
	Then $\theory{\Psi}{\gamma} \circ s = \theory{\Psi}{\hat{\gamma}} \colon S \rightarrow Q\Psi$. 
	If $\hat{s}_0 \colon 1 \rightarrow S$ satisfies
	$s \circ \hat{s}_0 = s_0$ then $(S,\hat{\gamma},\hat{s}_0)$ is correct w.r.t.~$\Psi$. 
\end{lemma}
We next define the crucial notion of \emph{counterexample}
to a pointed coalgebra: a subobject $\Psi'$ of $\Psi$ on which it is `incorrect'.
\begin{definition}
Let $(S,\Psi)$ be a table, and let $(S,\hat{\gamma},\hat{s}_0)$ be a pointed $B$-coalgebra on $S$.
Let $\Psi'$ be a subformula closed subobject of $\Phi$, such that
$\Psi$ is a subcoalgebra of $\Psi'$.
We say $\Psi'$ is a \emph{counterexample (for $(S,\hat{\gamma},\hat{s}_0)$, extending $\Psi$)} if 
$(S,\hat{\gamma},\hat{s}_0)$ is \emph{not} correct w.r.t.\ $\Psi'$.
\end{definition}
The following elementary lemma states that if there are no more counterexamples
for a coalgebra, then it is correct w.r.t.\ the object $\Phi$ of all formulas. 
\begin{lemma}\label{lm:no-more-counter}
	Let $(S,\Psi)$ be a table, and 
	let $(S,\hat{\gamma},\hat{s}_0)$ be a pointed $B$-coalgebra on $S$.
	Suppose that there are no counterexamples for 
	$(S,\hat{\gamma},\hat{s}_0)$ extending $\Psi$. Then
	$(S,\hat{\gamma},\hat{s}_0)$ is correct w.r.t.\ $\Phi$.  
\end{lemma}


The following describes, for a given table, how to extend it with the successors (in $X$)
of all states in $S$. As we will see below, by repeatedly applying this construction, 
one eventually obtains a closed table. 
\begin{definition}\label{def:closing}
    Let $(S,\Psi)$ be a sharp table. 
	Let $(\overline{S},q,r)$ be the (strong epi, mono)-factorisation 
	of the map $\theory{}{\gamma} \circ (s \vee \Gamma(s))$, as in 
	the diagram:  \vspace{-.3cm}
	$$
	\begin{tikzcd}
		S \vee \Gamma(S) \ar[r,"{s \vee \Gamma(s)}"] \ar[dr,"q"',twoheadrightarrow]
			& X \ar[r,"{\theory{}{\gamma}}"] 
			& \Q \Psi \\
			& \overline{S} \ar[ur,"r"',tail]
			& 
	\end{tikzcd}
	$$
	We define
	$\close(S,\Psi) ~{:=}~ \{ \overline{s} \colon \overline{S} \rightarrowtail X \mid
	\theory{}{\gamma} \circ \overline{s} = r, s \leq \overline{s} \leq s \vee \Gamma(s) \}
	$.
	For each $\overline{s} \in \close(S,\Psi)$ we have $s \leq \overline{s}$
	and thus $s = \overline{s} \circ \kappa$ for some $\kappa\colon S \to \overline{S}$.
	
	\end{definition}
       \begin{lemma}\label{lem:connecting}
         In Definition~\ref{def:closing}, for each $\overline{s} \in \close(S,\Psi)$,
        we have $\kappa = q \circ \inlv$.
       \end{lemma}
      We will refer to $\kappa=q \circ \inlv$ as the connecting map
	from $s$ to $\overline{s}$.

\begin{lemma}\label{lem:uglycondition} 
	In Definition~\ref{def:closing}, if there exists $q^{-1}  \colon \overline{S} \rightarrow S \vee \Gamma(S)$ 
	such that $q \circ q^{-1} = \id$ and $q^{-1} \circ q \circ \inlv = \inlv$, then 
	$\close(S,\Psi)$ is non-empty. 
\end{lemma}
By our assumptions, the hypothesis of Lemma~\ref{lem:uglycondition} is satisfied (Remark~\ref{rem:assumption-disc}),
hence $\close(S,\Psi)$ is non-empty. It is precisely (and only) at this point that we need
the strong condition about existence of right inverses to epimorphisms.

\subsection{The algorithm}\label{sec:alg}

Having defined closedness, counterexamples and 
a procedure for closing a table, we are ready to define the abstract algorithm. 
In the algorithm, the teacher has access to
a function $\counter((S,\hat{\gamma},\hat{s}_0),\Psi)$,
which returns the set of all counterexamples (extending $\Psi$) for the conjecture $(S,\hat{\gamma},\hat{s}_0)$.
If this set is empty, the coalgebra $(S,\hat{\gamma},\hat{s}_0)$ is correct (see Lemma~\ref{lm:no-more-counter}), 
otherwise the teacher picks one of its elements $\Psi'$. 
%
We also make use of $\close(S,\Psi)$, as given in Definition~\ref{def:closing}.

\begin{algorithm}[H]
\caption{Abstract learning algorithm}\label{alg:main}
\begin{algorithmic}[1]
\State $(S \stackrel{s}{\rightarrowtail} X) \gets (1 \stackrel{s_0}{\rightarrowtail} X)$
\State $\hat{s}_0 \gets \id_1$
\State $\Psi \gets 0$
\While {\texttt{true}}
	\While {$(S \stackrel{s}{\rightarrowtail} X,\Psi)$ is not closed} \label{ln:closing}
		\State let $(\overline{S} \stackrel{\overline{s}}{\rightarrowtail} X)
		\in \close(S,\Psi)$, with connecting map $\kappa \colon S \rightarrowtail \overline{S}$
		\State $(S \stackrel{s}{\rightarrowtail} X) \gets 
		(\overline{S} \stackrel{\overline{s}}{\rightarrowtail} X)$
		\State $\hat{s}_0 \gets \kappa \circ \hat{s}_0$
	\EndWhile
	\State let $(S,\hat{\gamma})$ be a conjecture for $(S,\Psi)$  \label{ln:conj}
	\If {$\counter((S,\hat{\gamma},\hat{s}_0),\Psi) = \emptyset$} 
		\State \textbf{return} $(S,\hat{\gamma}, \hat{s}_0)$ \label{ln:ret}
	\Else
		\State $\Psi \gets \Psi'$ for some $\Psi' \in \counter((S,\hat{\gamma},\hat{s}_0),\Psi)$ \label{ln:new-ctr}
	\EndIf
\EndWhile
\end{algorithmic}
\end{algorithm}
The algorithm takes as input the coalgebra $(X,\gamma,s_0)$ (which we fixed throughout this section). 
In every iteration of the outside loop, the table is first closed by repeatedly applying the procedure in Definition~\ref{def:closing}. Then, if the conjecture corresponding to the closed table is correct,
the algorithm returns it (Line~\ref{ln:ret}). Otherwise, a counterexample is
chosen (Line~\ref{ln:new-ctr}), and the algorithm continues. 

\subsection{Correctness and Termination}\label{sec:correctness-and-termination}

Correctness is stated in Theorem~\ref{thm:correctness}. 
It relies on establishing loop invariants: 

\begin{theorem}\label{thm:invariant}
The following is an invariant of both loops in Algorithm~\ref{alg:main}:
\begin{enumerate*}
	\item $(S,\Psi)$ is sharp, 
	\item $s \circ \hat{s}_0 = s_0$, and
	\item $s$ is $s_0$-prefix closed w.r.t.\ $\gamma$.
\end{enumerate*}
\end{theorem}

\begin{theorem}\label{thm:correctness}
	If Algorithm~\ref{alg:main} terminates, then it returns a pointed coalgebra 
	$(S,\hat{\gamma},\hat{s}_0)$ which is minimal w.r.t.\ logical equivalence, 
	reachable and correct w.r.t.~$\Phi$. 
\end{theorem}

%
%


In our termination arguments, we have to make an assumption about the coalgebra which
is to be learned. It does not need to be finite itself, but it should be finite
up to logical equivalence---in the case of deterministic automata, for instance,
this means the teacher has a (possibly infinite) automaton representing
a regular language. To speak about this precisely,
let $\Psi$ be a subobject of $\Phi$. We take a (strong epi, mono)-factorisation of the theory map, i.e., 
$
\theory{\Psi}{\gamma} = \left(
	\xymatrix{X \ar@{->>}[r]^-{e_\Psi}
		& \xvert{X}{\Psi}~ \ar@{>->}[r]^-{m_\Psi}
		& Q\Psi
	}
\right)
$ for some strong epi $e$ and mono $m$. 
We call the object $\xvert{X}{\Psi}$ in the middle the \emph{$\Psi$-logical quotient}.
For the termination result (Theorem~\ref{thm:termination}), $\xvert{X}{\Phi}$
is assumed to have finitely many quotients and subobjects, which
just amounts to finiteness, in $\Set$.

We start with termination of the inner while loop (Corollary~\ref{cor:term-inner}). This
relies on two results:
first, that once the connecting map $\kappa$ is an iso, the table is closed, 
and second, that---under a suitable assumption on the coalgebra $(X,\gamma)$---during execution
of the inner while loop, the map $\kappa$ will eventually be an iso. 

\begin{theorem}\label{thm:kappa-iso-closed}
	Let $(S, \Psi)$ be a sharp table, let $\overline{S} \in \close(S,\Psi)$
	and let $\kappa \colon S \rightarrow \overline{S}$ be the connecting map. 
	If $\kappa$ is an isomorphism, then $(S,\Psi)$ is closed. 
\end{theorem}

\begin{lemma}\label{lm:some-kappa-iso}
	Consider a sequence of sharp tables $(S_i \stackrel{s_i}{\rightarrowtail} X,\Psi)_{i \in \mathbb{N}}$
	such that $s_{i+1} \in \close(S_i, \Psi)$ for all $i$. Moreover,
	let $(\kappa_i \colon S_i \rightarrow S_{i+1})_{i \in \mathbb{N}}$ be the connecting maps (Definition~\ref{def:closing}).
	If the logical quotient $\xvert{X}{\Phi}$ of $X$ 
	has finitely many subobjects, then $\kappa_i$ is an isomorphism for some $i \in \mathbb{N}$. 
\end{lemma}

\begin{corollary}\label{cor:term-inner}
	If the $\Phi$-logical quotient $\xvert{X}{\Phi}$ has finitely many subobjects, then 
	the inner while loop of Algorithm~\ref{alg:main} terminates.
\end{corollary}

For the outer loop, 
we assume that $\xvert{X}{\Phi}$ has finitely many quotients,
ensuring that every sequence of 
counterexamples proposed by the teacher is finite. 

\begin{theorem}\label{thm:termination}
	If the $\Phi$-logical quotient $\xvert{X}{\Phi}$ has finitely many quotients
	and finitely many subobjects,
	then Algorithm~\ref{alg:main} terminates. 
\end{theorem}

\section{Future Work}
\label{sec:fw}

We showed how duality plays a natural role in automata learning,
through the central connection between states and tests. Based on this foundation, 
we proved correctness and termination of an abstract algorithm for coalgebra
learning. The generality is not so much in the base category (which, for the algorithm, we take to be $\Set$)
but rather in the functor used; we only require a few mild conditions on the functor,
and make no assumptions about its shape. The approach 
is thus considered \emph{coalgebra learning} rather
than automata learning. 

Returning to automata, an interesting direction is to extend the present
work to cover learning of, e.g., non-deterministic or alternating automata~\cite{BolligHKL09,AngluinEF15}
for a regular language. This would require explicitly handling
branching in the type of coalgebra. One promising direction would be to incorporate
the forgetful logics of~\cite{KlinR16}, which are defined within the same 
framework of coalgebraic logic as the current work. 
It is not difficult to define in this setting what it means for a table
to be closed `up to the branching part', stating, e.g., that even though
the table is not closed, all the successors of rows are present as combinations of other rows. 

Another approach would be to integrate monads into our framework, which are also
used to handle branching within the theory of coalgebras~\cite{JacobsSS15}. It is an intriguing question
whether the current approach, which allows to move beyond automata-like examples, 
can be combined with the CALF framework~\cite{heer17:lear}, which is very far in handling branching occuring
in various kinds of automata.


\paragraph{Acknowledgments.} We are grateful to Joshua Moerman, Nick Bezhanishvili, Gerco van Heerdt, Aleks Kissinger and Stefan Milius for valuable
discussions and suggestions. 

\bibliography{bibliography_learning_via_duality.bib}

\newpage

\begin{appendix}

\section{Proofs of Section~\ref{sec:base}}

\begin{proof}[Proof of Proposition~\ref{prop:existence-base}]
	Let $f \colon X \rightarrow B(Y)$. 
	Consider the collection of all pairs of maps $g_k \colon X \rightarrow B(U_k)$, 
	$m_k \colon U_k \rightarrow Y$ such that $B(m_k) \circ g_k = f_k$ and $m_k$ is a subobject,
	indexed by $k \in K$. Let 
	$
		m \colon \bigwedge \{m_k\}_{k \in K} \rightarrow Y
	$
	be the intersection of all the $m_k$ -- this is a (small) set since $\Cat{C}$ is well-powered.
	We abbreviate $\bigwedge \{m_k\}_{k \in K}$ by $I$.
	
	Since $B$ preserves intersections, $B(m) \colon B(I) \rightarrow B(Y)$ is the intersection 
	of all the subobjects $B(m_k)$.
Now the $g_k$'s form a cone over the $B(m_k)$'s, 
so we get a unique $g\colon X \rightarrow B(I)$ from the universal property of the pullback $B(I)$. 

We claim that $(I,g,m)$ is the base of $f$. To see this, 
first of all, note that $i$ is mono, and $B(m) \circ g = f$ by definition of $m$ and $g$. 
Further, if there is any $g' \colon X \rightarrow B(U), m' \colon U \rightarrow Y$ with
	$B(m') \circ g' = f$ and $m'$ monic 
	then it is (up to isomorphism) one of the $g_k$,$m_k$ pairs. 
	Hence, there is the map $x_k \colon I \rightarrow U_k$ in the limiting cone, i.e., $m_k \circ x_k = i$,
	and we have $B(x_k) \circ g = g_k$.
	Finally $x_k$ is unique among such maps, 
	since $B$ preserves monos (as it preserves intersections).
\end{proof}

\begin{proof}[Proof of Lemma~\ref{lm:nat-base}]
By our assumption on $Z$ there exists a morphism
$g: S \to \B Z$ such that $\B m \circ g = f$.
Therefore we have $ \B m' \circ \B e \circ g = \B h \circ f$
which shows that $m'$ is a candidate for the base of $\B h \circ f$.
We still need to check the universal property of the base.
To this aim let $g':S \to \B U$ and $n:U \to Y$ be the base
of $\B h \circ f$:
\begin{equation*}
\begin{tikzcd}
S \ar[rr, "f"] \ar[dr, "g"] \ar[ddrr, "g'"', bend right = 20] & &
\B X \ar[rr, "\B h"] & & \B Y \\
 & \B Z \ar[ur, "\B m"] \ar[rr, "\B e"] & & \B W \ar[ur, "\B m'"] \\
 & & \B U \ar[ur, "\B j"] \ar[uurr, "\B n"', bend right = 20] & & 
\end{tikzcd}
\end{equation*}
By the universal property of the base there is a morphism 
$j \colon U \rightarrow W$ making the lower right diagram commute.
Now, consider the following pullback:
\begin{equation*}
\begin{tikzcd}
P \ar[r, "p_n"] \ar[d, "p_h"'] \ar[dr, phantom, very near start,
"\lrcorner"] &
X \ar[d, "h"] \\
U \ar[r, "n"', tail] & Y 
\end{tikzcd}
\end{equation*}
This is a preimage because $n$ is mono and by assumption on $\B$
we have that this pullback is preserved under application of $\B$.
$S$ forms a cone over the diagram with $S \rightarrow \B X,
S \rightarrow \B U$. So there exists a map from $S$ to the pullback.
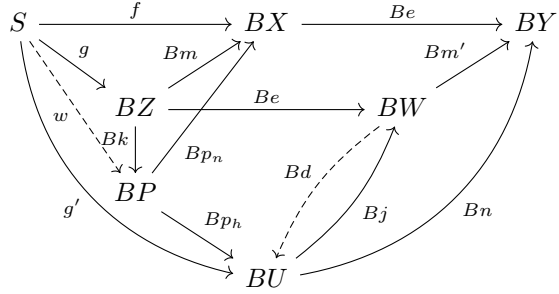
\begin{equation*}
\begin{tikzcd}
S \ar[rr, "f"] \ar[dr, "g"] \ar[dddrr, "g'"', bend right = 35]
\ar[ddr, "w"', dashed] & &
\B X \ar[rr, "\B e"] & & \B Y \\
 & \B Z \ar[ur, "\B m"] \ar[rr, "\B e"] \ar[d, "\B k"', near start] & &
\B W \ar[ur, "\B m'"] \ar[ddl, "\B d"', dashed, bend right = 15] \\
 & \B P \ar[uur, "\B p_n"', near start] \ar[dr, "\B p_h"] &
 & & \\
 & & \B U \ar[uur, "\B j"', bend right = 15]
 \ar[uuurr, "\B n"', bend right = 35] & & 
\end{tikzcd}
\end{equation*}
$p_n$ is mono, so $\B p_n \circ w$  is a base factorization. So, we get an arrow
from $\B Z$ to $\B P$, i.e., $k \colon Z \rightarrow P$ such that:
\begin{itemize}
\item[(i)] $\B k \circ g = w$
\item[(ii)] $p_n \circ k = m$
\end{itemize}
Consider the following diagram. According to the diagonal
filling property, we have $\exists ! d \colon W \rightarrow U$
such that (i) $n \circ d = m'$ and (ii) $d \circ e = p_h \circ k$.
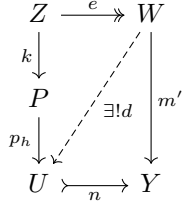
\begin{equation*}
\begin{tikzcd}
Z \ar[r, "e", twoheadrightarrow] \ar[d, "k"'] &
W \ar[dd, "m'"] \ar[ddl, "\exists!d", dashed] \\
P \ar[d, "p_h"'] & \\
U \ar[r, "n"', tail] & Y
\end{tikzcd}
\end{equation*}
By the universal property of the base we have $d \circ j =
\mathit{id}_U$. Moreover, $m' \circ j \circ d = n \circ d = m'$, and
because $m'$ is monic we have $j \circ d = \mathit{id}_W$.
\end{proof}

\begin{lemma}\label{lattice1}
If $\Cat{C}$ is complete and well-powered, then for each $X \in \Cat{C}$, we have $\Sub(X)$ has arbitrary meets. 
Consequently, $\Sub(X)$ is a complete lattice.
\end{lemma}
\begin{proof}[Proof of Lemma~\ref{lattice1}]
Consider some $X \in \Cat{C}$ and an arbitrary family of subobjects $\{m_i \colon S_1 \rightarrow X\}_{i \in I}$. Let $P$
be the pullback with pullback maps $\{p_1 \colon P \rightarrow S_i\}_{i \in I}$. As the $m_i$ are mono, the
$p_i$ are mono as well, so let define $p := m_i \circ p_i \colon P \rightarrow X \in \Sub(X)$. Obviously, we have
$p \leq m_i \forall i \in I$. So, $P$ is a lower bound. To see that $P$ is the greatest lower bound, consider an arbitrary
$P'$ with $p' \colon P' \rightarrow X$ that is a lower bound of the same family of subobjects. By definition of lower bound
we have for each $i \in I$ a map $p_i' \colon P' \rightarrow S_i$ s.t. $m_i \circ p_i' = p$. By universal property of
the (wide) pullback, there exists a unique map $c \colon P' \rightarrow P$ s.t. $p \circ c = p'$, i.e., 
$p' \leq p$. As $P'$ was an arbitrary lower bound, we showed that $P$ is the greatest lower bound.
This finishes the proof of the fact that $\Sub(X)$ has arbitrary meets. Completeness of the lattice can now be proven
in a standard way by defining the join of an arbitrary collection of subobjects as the meet of all upper bounds of this collection.
\end{proof}

\begin{proof}[Proof of Lemma~\ref{lm:gamma-monotone}]
  It is obvious that $\Gamma$ is well-defined. To check monotonicity, we consider the following diagram for subobjects $s:S \to X$ and $s':S' \to X$
  such that $s \leq s'$.
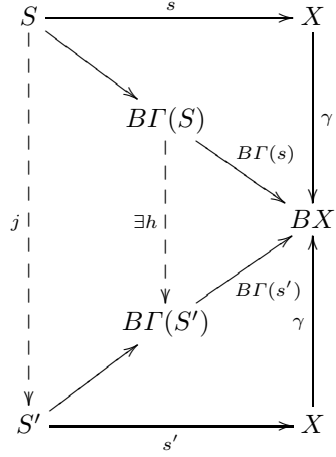
  \[
   \xymatrix{S \ar@{-->}[dddd]_{j} \ar[rd] \ar[rr]^{s} & & X \ar[dd]^\gamma \\
      & \B \Gamma(S) \ar[rd]^{\B \Gamma(s)} \ar@{-->}[dd]_{\exists h} & \\
      & & \B X \\
      & \B \Gamma(S') \ar[ru]_{\B \Gamma(s')} &  \\
      S' \ar[ru] \ar[rr]_{s'} & & X \ar[uu]^\gamma }
  \]
  Here $j$ exists by the definition of $\leq$ and $h$ exists by the universal property of the base of $\gamma \circ s$.
  Therefore we have $\Gamma(s) \leq \Gamma(s')$ as required.
\end{proof}

\begin{proof}[Proof of Proposition~\ref{prop:subcoalg-base}]
\begin{itemize}
\item[$\Rightarrow$] Consider the following diagram 
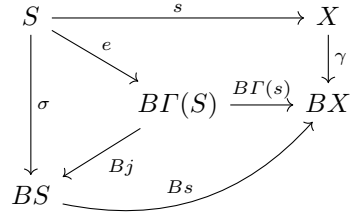
\begin{equation*}
\begin{tikzcd}
S \ar[rr, "s"] \ar[dd, "\sigma"] \ar[dr, "e"] & & X \ar[d, "\gamma"] \\
 & \B \Gamma(S) \ar[r, "\B \Gamma(s)"] \ar[dl,"\B j"] & \B X \\
 \B S \ar[urr, "\B s", bend right] & & \\
\end{tikzcd}
\end{equation*}
As $s$ is a subcoalgebra there exists $\sigma \colon S \rightarrow \B S$ s.t. the outer square commutes. By the universal
property of the base there exists $j \colon \Gamma(S) \rightarrow S$ s.t. $s \circ j = \Gamma(s)$. In other words,
$\Gamma(s) \leq s$ as required.
\item[$\Leftarrow$] By assumption there exists a $j: \Gamma(S) \to S$ such that 
$s \circ j = \Gamma(s)$. We define a $\B$-coalgebra structure on $S$ by putting  $\sigma :=  j \circ e$. We have to show that the outer square in the above diagram  
commutes, but this is easy
to show because the inner square commutes by definition of the base, the left triangle commutes by definition of 
$\sigma$ and the right one by assumption on $j$.
\end{itemize}
Finally, that the collection of subcoalgebras of $(X,\gamma)$ forms a  complete lattice is now a direct consequence of
the fact that the collection of pre-fixpoints of the monotone operator $\Gamma$ forms a complete lattice (Knaster-Tarski theorem).
\end{proof}

\begin{proof}[Proof of Theorem~\ref{thm:reach-base}]
	Suppose $X \cong \lfp(\Gamma \vee x_0)$, and let $s \colon S \rightarrowtail X$
	be a subcoalgebra together with an arrow $s_0 \colon 1 \to S$ with $x_0 = s \circ s_0$. 
	The latter implies that $x_0 \leq s$. Further, since $S$ is a subcoalgebra, 
	by Proposition~\ref{prop:subcoalg-base} we get $\Gamma(s) \leq s$. 
	Hence $\Gamma(s) \vee x_0 \leq s$, i.e., $s$ is a pre-fixed point of $\Gamma \vee x_0$. 
	By the Knaster-Tarski theorem (using that $\Sub(X)$ is a complete lattice),
	$\lfp(\Gamma \vee x_0)$ is the least pre-fixed point, so it now suffices to prove that $s \leq \lfp(\Gamma \vee x_0)$.
	But this follows easily, since $s$ is a subobject of $X$ and $X \cong \lfp(\Gamma \vee x_0)$.
	
	Conversely, suppose  $(X,\gamma,x_0)$ is reachable. 	
	We have that $\Gamma(\lfp(\Gamma \vee x_0)) \leq
	\Gamma(\lfp(\Gamma \vee x_0)) \vee x_0 = \lfp(\Gamma \vee x_0)$, so by Proposition~\ref{prop:subcoalg-base},
	$\lfp(\Gamma \vee x_0)$ is a subcoalgebra of $(X, \gamma)$. 
	Moreover, we have
	$x_0 \leq \Gamma(\lfp(\Gamma \vee x_0)) \vee x_0 = \lfp(\Gamma \vee x_0)$, so there
	exists a map $s_0 \colon 1 \rightarrow \lfp(\Gamma \vee x_0)$ such that $x_0 = s \circ s_0$,
	where $s \colon \lfp(\Gamma \vee x_0) \rightarrowtail X$ is the inclusion. Hence,
	by definition of reachability, we get that $s$ is an isomorphism. 	
	\qed
\end{proof}

\section{Proofs of Section~\ref{sec:learning}}

\begin{proof}[Proof of Theorem~\ref{thm:conjecture}]
	Given a table $(S,\Psi)$ that is closed, it is straightforward to construct
	a conjecture $\hat{\gamma}$ as composite of $g:S \to \B \Gamma(S)$ 
	and the arrow $B k: \B\Gamma(S) \to S$, where $g$ is part of the 
	base $(\Gamma(S),g,\Gamma(s))$ of $\gamma \circ s$ and $k: \Gamma(S) \to S$
	is the morphism that exists by closedness of $(S,\Psi)$.
    
	For the converse, consider a conjecture $(S, \hat{\gamma})$ for a sharp table  $(S,\Psi)$,
	let $(\Gamma(S),g,\Gamma(s))$ be the base of $\gamma \circ s$ and let $(h: \Gamma(S) \to Y, m:Y \to Q \Psi)$ be the factorisation of
	$\theory{}{\gamma} \circ \Gamma(s)$.
	By Lemma~\ref{lm:nat-base}, as $h$ is epi,  we have that $(Y, Bh \circ g, m)$ is the base of $B\theory{}{\gamma} \circ \gamma \circ s$. The situation
	is depicted in the (commuting) upper square of the diagram below). 
\begin{equation}\label{eq:extended-base}
	\begin{tikzcd}
		S \ar[r,"s",tail] \ar[dr,"g"] \ar[ddr,"\hat{\gamma}"',bend right = 20]
			& X \ar[r,"\gamma"]
			& B{X} \ar[r,"{B\theory{}{\gamma}}"]
			& BQ(\Psi)\\
			& B(\Gamma(S)) \ar[r,"Bh"]
			& B(Y) \ar[ur,"Bm"] \ar[dl,"Bj"]
			& \\
			& B{S} \ar[r,"Bs"']
			& B{X} \ar[uur,"{B\theory{}{\gamma}}"',bend right = 20]
			&
	\end{tikzcd}
\end{equation}
The bigger outer square commutes by the fact that $\hat \gamma$ is assumed to be a conjecture.
As the table is sharp, we have $\theory{}{\gamma} \circ s$ is mono. Therefore the universal property of the base 
yields existence of a morphism $j: Y \to S$ such that $\theory{}{\gamma} \circ s \circ j = m$.
We define $k \mathrel{:=} j \circ h$ and claim that this $k$ is a witness for $(S,\Psi)$, i.e., that $k$ makes the
relevant diagram from Definition~\ref{def:closed} commute. To see this, we calculate:
\[
      \theory{}{\gamma} \circ \Gamma(s) = m \circ h =  \theory{}{\gamma} \circ s \circ j  \circ h =  \theory{}{\gamma} \circ s \circ k .
\]
This finishes the proof.
\qed
\end{proof}

\begin{proof}[Proof of Lemma~\ref{lm:truth-lemma-tables}]	
	The map $\theory{}{\hat{\gamma}}$ is, by definition, the unique map making the following diagram commute.
	$$
	\begin{tikzcd}
		S \ar[rr, "{\theory{}{\hat{\gamma}}}"] \ar[d,"{\hat{\gamma}}"']
			& & \Q \Psi \\
		B{S} \ar[r, "{B\theory{}{\hat{\gamma}}}"']
			& B \Q \Psi \ar[r, "{\rho^\flat_{\Psi}}"'] 
			& \Q L \Psi \ar[u,"d"']
	\end{tikzcd}
	$$
	where $d \colon \Psi \rightarrow L\Psi$ is the coalgebra structure from subformula closedness
	of $\Psi$. 
	Consider the following diagram: 
	\begin{equation*}
	\begin{tikzcd}
	S \ar[r, "s", tail] \ar[d, "\hat{\gamma}"'] &
	X \ar[d, "\gamma"] \ar[rr,"{\theory{}{\gamma}}"] 
	& & \Q \Psi \\
	B{S} \ar[r, "Bs"] &
	B{X} \ar[r, "{B\theory{}{\hat{\gamma}}}"'] &
	B \Q \Psi \ar[r, "{\rho^\flat_{\Psi}}"'] &
	\Q L \Psi  \ar[u,"d"']
	\end{tikzcd}
	\end{equation*}
	The rectangle on the right commutes by definition of $\theory{\Psi}{\gamma}$.
	Together with $\hat{\gamma}$
	being a conjecture, it follows that the outside of the diagram commutes. 
	Since $\theory{\Psi}{\hat{\gamma}}$ is the unique such map,
	we have $\theory{\Psi}{\gamma} \circ s = \theory{\Psi}{\gamma}$.
	\qed
\end{proof}

\begin{proof}[Proof of Lemma~\ref{lm:no-more-counter}]
If there is no counterexample,
	then in particular $\Phi$ is not a counterexample. 
	The object $\Phi$ is subformula-closed subobject of itself,
	and $\Psi$ is a subcoalgebra of $\Phi$. Hence, by definition of
	counterexamples, it must be the case that $(S,\hat{\gamma},\hat{s}_0)$ correct w.r.t.\ $\Phi$.
	\qed
\end{proof}

\begin{proof}[Proof of Lemma~\ref{lem:connecting}]
Let $\overline{s} \in \close(S,\Psi)$. We calculate:
\begin{eqnarray*}
	 \theory{}{\gamma} \circ \overline{s} \circ \kappa & = & \theory{}{\gamma} \circ s = \theory{}{\gamma} \circ (s \vee \Gamma(s)) \circ \inlv = r \circ q \circ \inlv \\
		 & = & \theory{}{\gamma} \circ \overline{s} \circ q \circ \inlv 
\end{eqnarray*}
which implies $\kappa = q \circ \inlv$ as $r =  \theory{}{\gamma} \circ \overline{s}$ is a mono.
\end{proof}

\begin{proof}[Proof of Lemma~\ref{lem:uglycondition}]
    Given the assumption of the lemma we are able to define a morphism $\overline{s}:\overline{S} \to X$ 
    by putting $\overline{s} \mathrel{:=} (s \vee \Gamma(s)) \circ q^{-1}$. Obviously $\overline{s}$ is a mono
    as it is defined as composition of monos. Furthermore, by definition, we have $\overline{s} \leq s \vee \Gamma(s)$.
    To see that $s \leq \overline{s}$, we calculate
    \[
      \overline{s} \circ q \circ \inlv =  (s \vee \Gamma(s)) \circ q^{-1} \circ q \circ \inlv =   (s \vee \Gamma(s)) \circ \inlv = s .
    \]
   Finally, the condition concerning the theory map also follows easily:
   \[
    \theory{}{\gamma} \circ \overline{s} = \theory{}{\gamma} \circ (s \vee \Gamma(s)) \circ q^{-1} = r \circ q \circ q^{-1} = r .
   \]
   This finishes the proof of the lemma.
\end{proof}

We need a few auxiliary lemma's in the proofs below.

\begin{lemma}\label{lm:close-pres-sharp}
	If $(S,\Psi)$ is sharp, then $(\overline{s} \colon 
	\overline{S} \rightarrow X,\Psi)$ is sharp
	for any $\overline{s} \in \close(S,\Psi)$.
\end{lemma}
\begin{proof} 
 	This folllows immediately from Definition~\ref{def:closing}, 
 	since $\theory{}{\gamma} \circ \overline{s} = r$, where $r$ is monic.
 	\qed
\end{proof}

\begin{lemma}\label{lm:counter-theory}
	Let $\Psi$ and $\Psi'$ be subformula closed, with $\Psi$ a subcoalgebra
	of $\Psi'$, witnessed by a mono $i \colon \Psi \rightarrowtail \Psi'$.
	Then we have $\Q i \circ \theory{\Psi'}{\gamma} = \theory{\Psi}{\gamma}$. 
\end{lemma}
\begin{proof}[Proof of Lemma~\ref{lm:counter-theory}]
	Let $\sigma' \colon \Psi' \rightarrow L\Psi'$ and $\sigma \colon \Psi \rightarrow L\Psi$ be 
	the coalgebra structures from subformula closedness of $\Psi'$ and $\Psi$ respectively. 
	Consider the following diagram.
	$$
	\begin{tikzcd}
		X \ar[rr, "{\theory{\Psi'}{\gamma}}"] \ar[d,"\gamma"']
			& & \Q(\Psi') \ar[r, "{\Q i}"]
			& \Q \Psi \\
		B{X} \ar[r, "B{\theory{\Psi'}{\gamma}}"]
			& B \Q(\Psi') \ar[r, "{\rho^\flat_{\Psi'}}"] \ar[dr, "{B\Q i}"']
			& \Q L(\Psi')  \ar[u,"{\sigma'}"'] \ar[r, "{\Q Li}"]
			& \Q L \Psi  \ar[u,"\sigma"'] \\
			& & B \Q \Psi \ar[ur, "{\rho^\flat_{\Psi}}"']
	\end{tikzcd}
	$$
	By definition, $\theory{\Psi'}{\gamma}$ is the unique map making the left rectagle
	commute. The (right) square commutes by assumption that $i$ is a coalgebra
	homomorphism, and the (lower) triangle by naturality. 
	Since $\theory{\Psi}{\gamma}$ is the unique map 
	such that $\theory{\Psi}{\gamma} = \Q \sigma \circ \rho^\flat_{\Psi} \circ B\theory{\Psi}{\gamma} \circ \gamma$,
	we have $\theory{\Psi}{\gamma} = \Q i \circ \theory{\Psi'}{\gamma}$. 
	\qed
\end{proof}

\begin{lemma}\label{lm:counter-sharp}
	Suppose $(S,\Psi)$ is sharp, and $\Psi'$ is a counterexample. Then 
	$(S,\Psi')$ is again sharp. 
\end{lemma}
\begin{proof} 
	Let $i \colon \Psi \rightarrowtail \Psi'$ be the inclusion of the coalgebra $(\Psi,\sigma)$ into $(\Psi',\sigma')$.
	By Lemma~\ref{lm:counter-theory}, we have 
	$\Q i \circ \theory{\Psi'}{\gamma} = \theory{\Psi}{\gamma}$. 
	Hence $\Q i \circ \theory{\Psi'}{\gamma} \circ s = \theory{\Psi}{\gamma} \circ s$,
	and since $\theory{\Psi}{\gamma} \circ s$ is monic, it follows that 
	$\theory{\Psi'}{\gamma} \circ s$ is monic. 
	\qed
\end{proof}

\begin{proof}[Proof of Theorem~\ref{thm:invariant}]
	To show that each of these is an invariant of both loops, it suffices
	to prove that they hold at once we enter the first iteration of the outer loop,
	and that both loops preserve them (that it holds at the start of each 
	first iteration of the inner loop then follows). 
	\begin{enumerate}
		
		\item 
		(Holds at entry of the outer loop.) At this point, $(S,\Psi) = (S_0, 0)$. 
		Since $\Q$ is a right adjoint, it maps $0$ to the terminal object $\Q 0 = 1$ of $\Cat{C}$. 
		Hence, the map from $S_0=1$ to $\Q \Psi$ is of the form 
		\begin{tikzcd}
		1 \ar[r, "s_0"] &
		X \ar[r, "{\theory{}{\gamma}}"] &
	 	\Q \Psi = 1 \end{tikzcd}
	 	which is an iso, so in particular monic. 
	 	
	 	(Preserved by the inner loop.) This follows from Lemma~\ref{lm:close-pres-sharp}.
	 	
	 	(Preserved by the loop body.) 
	 	If $(S,\Psi)$ is sharp on entry of the body of the outer loop, 
	 	and the inner loop terminates, then $(S,\Psi)$
	 	is again sharp at Line~\ref{ln:conj}. It only remains to show that if $\Psi'$ is
	 	a counterexample (extending $\Psi'$) for a conjecture for $(S,\Psi)$, then $(S,\Psi')$ is 
	 	sharp. This follows, in turn, from Lemma~\ref{lm:counter-sharp}.

		\item (Holds at entry of the outer loop.) Follows immediately
		from the first two lines of the algorithm. 
		
		(Preserved by the inner loop.) Suppose $\hat{s}_0 \circ s = s_0$,
		and let $\overline{s} \in \close(S,\Psi)$. We need to prove that
		$\overline{s}_0 \circ \kappa \circ \hat{s}_0 = s_0$ where $\kappa:S \to \overline{S}$
		is the connecting map.
		Indeed, we have $\overline{s}_0 \circ \kappa \circ \hat{s}_0 = s \circ \hat{s}_0 = s_0$,
		by definition of $\kappa$ and assumption, respectively. 
		
		(Preserved by the outer loop.) This follows immediately from preservation by the inner loop.
		
		\item Clearly the initial configuration $(S_0,0)$ is $s_0$-prefix closed. 
		  Suppose now that $(S,\Psi)$ is a table with $s$ being $s_0$-prefix closed. 
		  We need to check that any $\overline{s} \in \close(S,\Psi)$ is $s_0$-prefix closed as well.
		  By assumption on $(S,\Psi)$ we have
		  $s = \bigvee_{i=0}^n s_i$ for a suitable family of subobjects $s_0,\ldots,s_n$.
		  Let $\overline{s} \in  \close(S,\Psi)$. Then by definition we have
		  $\overline{s} \leq s \vee \Gamma(s)$, so we put $s_{n+1} \mathrel{:=} \Gamma(s) \wedge \overline{s}$. 
		  It is then easy to check that $\overline{s}$ is $s_0$-prefixclosed:
		  \[ \bigvee_{i=0}^{n+1} s_i = \bigvee_{i=0}^n s_i \vee s_{n+1} = s \vee ( \Gamma(s) \wedge \overline{s}) = (s \vee \Gamma(s))\wedge  (s\vee \overline{s})   = \overline{s}\]
		where the last equality follows from  $s \leq \overline{s} \leq s \vee \Gamma(s)$ . By definition
	       we have $s_{n+1} \leq \Gamma(s) =  \Gamma(\bigvee_{i=0}^n s_i)$ as required.
		\end{enumerate}
\end{proof}

\begin{proof}[Proof of Theorem~\ref{thm:correctness}]
  Minimality w.r.t logical equivalence follows from the fact that sharpness of the table is
  maintained throughout. As the algorithm terminated 
  there is no counterexample, which means by Lemma~\ref{lm:no-more-counter} that the coalgebra is correct w.r.t. $\Phi$. 
  For reachability we show that the pointed coalgebra
  that is returned by the algorithm is reachable by showing that
  {\em any} conjecture that is constructed during the run of the algorithm is reachable.
  While running the algorithm we will only encounter conjectures that are built from
  tables that are both sharp and closed.
  Therefore we consider an arbitrary sharp and closed table $(S,\Psi)$ together with the conjecture  $(S,\hat{\gamma})$  that exists
  according to Theorem~\ref{thm:conjecture}. We are going to prove that $(S,\hat{\gamma},\hat{s_0})$ is reachable. 
  
  By Theorem~\ref{thm:invariant} we
  know that $(S,\Psi)$ is $s_0$-prefix closed. This means that 
  $s = \bigvee_{i=0}^n s_i$ for suitable subobjects $s_1, \ldots,s_n \in \Sub(X)$.
  Suppose now that $(\overline{S},\overline{\gamma},\overline{s_0})$ is
  a subcoalgebra of $(S,\hat{\gamma},\hat{s_0})$ with inclusion $j:\overline{S} \to S$
  such that $j \circ \overline{s_0} = \hat{s_0}$.
  We prove by induction on $i$ that $s_i \leq \overline{s}$ for all $i \in \{0,\ldots,n\}$
  and thus $s \leq \overline{s}$ - this will imply $s = \overline{s}$ and thus, as $\overline{s}$ was
  assumed to be an arbitrary (pointed) subcoalgebra, reachability of $(S,\hat{\gamma},\hat{s_0})$.
  \\
  {\noindent \em Case} $i=0$. Then $\overline{s} \circ \overline{s_0} = s  \circ j \circ \overline{s_0} = s \circ \hat{s_0} = s_0$
  and thus $s_0 \leq \overline{s}$ as required. \\
  {\noindent \em Case} $i=j+1$. Then 
  \[
      \theory{}{\gamma} \circ s_{j+1} \leq  \theory{}{\gamma} \circ \Gamma(\bigvee_{i=0}^j s_j) 
           \stackrel{\mbox{\tiny I.H.}}{\leq}  \theory{}{\gamma} \circ \Gamma(\overline{s}) 
          \stackrel{\mbox{\tiny Thm.~\ref{thm:conjecture}}}{\leq} 
            \theory{}{\gamma} \circ \overline{s}
  \]
  where we slightly abuse notation by writing $f \leq g$ for arbitrary morphisms $f:X_1 \to Y$ and $g:X_2 \to Y$ is 
  there exists a morphism $m:X_1 \to X_2$ such that $g \circ m = f$.
  The inequality implies that there is a map $k_{j+1}: S_i \to \overline{S}$ such that 
  $ \theory{}{\gamma} \circ \overline{s} \circ k_{j+1} = \theory{}{\gamma} \circ s_{j+1}$.
  This implies 
  \begin{equation}\label{eq:smallermodtheory}
     \theory{}{\gamma} \circ s \circ j \circ k_{j+1} = \theory{}{\gamma} \circ s_{j+1}
  \end{equation}
  On the other hand, we have $s \circ \mathrm{in}_{j+1} = s_{j+1}$ where $\mathrm{in}_{j+1} = s_{j+1}$
  denotes the inclusion of $s_{j+1}$ into $s$.
  Therefore we have $ \theory{}{\gamma} \circ s \circ \mathrm{in}_{j+1}  = \theory{}{\gamma} \circ s_{j+1}$.
  Together with~(\ref{eq:smallermodtheory}) this implies $ \theory{}{\gamma} \circ s \circ \mathrm{in}_{j+1} =  \theory{}{\gamma} \circ s \circ j \circ k_{j+1}$.
  By sharpness of the table $s$ we obtain $ \mathrm{in}_{j+1} =  j \circ k_{j+1}$ and finally
  $s_{j+1} = s \circ  \mathrm{in}_{j+1} = s \circ j \circ k_{j+1} = s' \circ k_{j+1}$ which shows that
  $s_{j+1} \leq s'$. This finishes the induction proof.
  \qed
\end{proof}

\begin{proof}[Proof of Theorem~\ref{thm:kappa-iso-closed}]
	Suppose $\kappa$ is an isomorphism. Consider the following diagram (using
	the notation from Definition~\ref{def:closing}),
	where $g \colon S \rightarrow B{X}$ is the map which forms
	the base of $\gamma \circ s$ together with $\Gamma(s) \colon \Gamma(S) \rightarrow X$. 
	
\begin{equation*}
\begin{tikzcd}
S \ar[d, "g"'] \ar[r,"s"] &
X \ar[r,"\gamma"] &
B{X} \ar[dddd, "B\theory{}{\gamma}"] \\
B(\Gamma(S)) \ar[d,"{B\inrv}"'] \ar[urr,"B(\Gamma(s))"] \\
B(S \vee \Gamma(S)) \ar[d,"{Bq}"'] \ar[uurr,"{B (s \vee \Gamma(s))}"']\\
B(\overline{S}) \ar[d,"{B\kappa^{-1}}"'] \ar[drr,"Br"] \\
B{S} \ar[r, "Bs"'] &
B{X} \ar[r, "B\theory{}{\gamma}"'] &
BQ(\Psi)
\end{tikzcd}
\end{equation*}
The inner shapes commute, from top to bottom: (1) by definition the base,
(2) by definition of $\inrv$, (3) by definition of $(q,r)$; for the bottom triangle (4), 
we have 
$$
\theory{}{\gamma} \circ s \stackrel{\mbox{\tiny Def. of $\kappa$}}{=} \theory{}{\gamma} \circ \overline{s} \circ \kappa    \stackrel{\mbox{\tiny $\overline{s} \in \close(S,\Psi)$}}{=} 
r \circ \kappa
$$
which suffices since $\kappa$ is an iso. Since the entire diagram commutes,
the coalgebra structure on $S$ gives a conjecture for $(S,\Psi)$. Hence, by Theorem~\ref{thm:conjecture}, 
the table is closed. \qed
\end{proof}

\begin{proof}[Proof of Lemma~\ref{lm:some-kappa-iso}]
	First, observe that the $S_i$'s form an increasing chain of subobjects of $X$.
	Since all these tables $(S_i,\Psi)$ are sharp, they give rise to an increasing chain of subobjects of 
	$Q(\Psi)$, by composition with $\theory{\Psi}{\gamma}$,
	given by $\theory{\Psi}{\gamma} \circ s_i \colon S_i \rightarrow Q(\Psi)$.
	By Lemma~\ref{lm:counter-theory},
	it follows that each $\theory{\Phi}{\gamma} \circ s_i \colon S_i \rightarrow Q \Phi$
	is monic, and we obtain a sequence of subobjects of $\Phi$:
	$$
	\begin{tikzcd}
		S_0 \ar[r,"\kappa_0"] \ar[d,"s_0"]
			& S_1 \ar[r,"\kappa_1"] \ar[d,"s_1"]
			& S_2 \ar[r,"\kappa_2"] \ar[d,"s_2"]
			& \ldots \\
		X \ar[dr,"{\theory{\Phi}{\gamma}}"']
			& X \ar[d,"{\theory{\Phi}{\gamma}}"]
			& X \ar[dl,"{\theory{\Phi}{\gamma}}"] 
			& \ldots \\
			& Q \Phi & & 
	\end{tikzcd}
	$$
	It follows that this induces a chain of subobjects of $\xvert{X}{\Phi}$:
		$$
	\begin{tikzcd}
		S_0 \ar[r,"\kappa_0"] \ar[d,"s_0"]
			& S_1 \ar[r,"\kappa_1"] \ar[d,"s_1"]
			& S_2 \ar[r,"\kappa_2"] \ar[d,"s_2"]
			& \ldots \\
		X \ar[dr,"{e_\Phi}"',twoheadrightarrow]
			& X \ar[d,"{e_\Phi}",twoheadrightarrow]
			& X \ar[dl,"{e_\Phi}",twoheadrightarrow] 
			& \ldots \\
			& \xvert{X}{\Phi} & & 
	\end{tikzcd}
	$$
	By assumption, $\xvert{X}{\Phi}$ has finitely many subobjects, so $\kappa_i$
	must be an isomorphism for some $i$. 
	\qed
\end{proof}

\begin{proof}[Proof of Corollary~\ref{cor:term-inner}]
	The while loop computes a chain of subobjects of $X$ as in Lemma~\ref{lm:some-kappa-iso};
	in particular, each of these forms a sharp table (with $\Psi$),
	since sharpness is an invariant (Theorem~\ref{thm:invariant}). 
	Hence, after a finite number of iterations, $\kappa$ is an iso. 
	By Theorem~\ref{thm:kappa-iso-closed} this implies that $(S,\Psi)$ is closed,
	which means the guard of the while loop is false.
	\qed
\end{proof}

For termination of the outer loop, we need several auxiliary lemmas.
\begin{lemma}\label{lm:subform-coalg}
	Let $(S,\Psi)$ be table, and let $\Psi'$ 
	be a subformula-closed subobject of $\Phi$, such that
	$\Psi$ is a subcoalgebra of $\Psi'$.
	Then there is a unique map $q$ making the following diagram commute:
\begin{equation}\label{eq:psi-quotient}
\begin{tikzcd}
	X \ar[r, "{e_{\Psi'}}", twoheadrightarrow] \ar[dr, "{e_{\Psi}}"', bend right=20, twoheadrightarrow] 
		& \xvert{X}{\Psi'} \ar[r, "{m_{\Psi'}}", tail] \ar[d, "q", dashed, twoheadrightarrow]
		& \Q \Psi' \ar[d, "\Q i"] \\
	& \xvert{X}{\Psi} \ar[r, "{m_{\Psi}}"', tail]
	& \Q \Psi
\end{tikzcd}
\end{equation}
	Moreover, this map $q$ is an epimorphism.
\end{lemma}
\begin{proof}
	The outside of the diagram commutes by Lemma~\ref{lm:counter-theory}. The map
	$q$ arises by the unique fill-in property. That $q$ is an epi follows
	since $e_{\Psi}$ is an epi, and $e_{\Psi} = q \circ e_{\Psi'}$.
	\qed
\end{proof}

\begin{lemma}\label{lm:counterexample-iso}
	Let $(S \stackrel{s}{\rightarrowtail} X,\Psi)$ be a closed table, and $(S,\hat{\gamma}, \hat{s}_0)$
	a pointed coalgebra, such that $(S,\hat{\gamma})$ is a conjecture
	and $s \circ \hat{s}_0 = s_0$. 
	If $\Psi'$ is a counterexample for $(S,\hat{\gamma},\hat{s}_0)$, then
	the map $q \colon \xvert{X}{\Psi'} \rightarrow \xvert{X}{\Psi}$
	from Lemma~\ref{lm:subform-coalg} is \emph{not} an isomorphism. 
\end{lemma}
\begin{proof}
	Suppose that $q$ is an iso; we prove that, in that case, 
	$(S,\hat{\gamma},\hat{s}_0)$ is correct w.r.t.\
	$\Psi'$. 
	Let $q^{-1}$ be the inverse of $q$. Since 
	$q \circ e_{\Psi'} = e_\Psi$ we also have $e_{\Psi'} = q^{-1} \circ e_\Psi$. Hence,
	the two shapes on the lower right in the following diagram commute:
	\begin{equation*}
	\begin{tikzcd}
	S \ar[r, "s", tail] \ar[d, "\hat{\gamma}"'] &
		X \ar[r, "\gamma"] &
		B{X} \ar[d, "Be_{\Psi}"'] \ar[ddr, "{Be_{\Psi'}}", bend left=20]
		&
		\\
	B{S} \ar[r, "Bs"'] &
	B{X} \ar[r, "Be_{\Psi}"] \ar[drr, "{Be_{\Psi'}}"', bend right=20] &
	B(\xvert{X}{\Psi}) \ar[dr, "{Bq^{-1}}"'] &
		\\
		& & & B(\xvert{X}{\Psi'})
	\end{tikzcd}
	\end{equation*}	
	The rectangle commutes since $\hat{\gamma}$ is a conjecture for the 
	closed table $(S,\Psi)$. Since the entire diagram commutes,
	it shows that $(S,\hat{\gamma})$ is a conjecture for the closed table 
	$(S,\Psi')$ as well. 
	Together with $s \circ \hat{s}_0 = s_0$, by Lemma~\ref{lm:truth-lemma-tables},
	we obtain that $(S,\hat{\gamma},\hat{s}_0)$ is correct w.r.t.\ $\Psi'$. 
		\qed
\end{proof}
\begin{proof}[Proof of Theorem~\ref{thm:termination}]
	The inner while loop terminates in each iteration 
	of the outer loop by Corollary~\ref{cor:term-inner}.
	The outer loop generates a sequence 
	$\Psi_0, \Psi_1, \Psi_2, \ldots$ of subobjects, such that for each $i$,
	there is a pointed coalgebra $(S_i, \hat{\gamma}, \hat{s}_0)$ such that
	\begin{itemize}
		\item $(S_i \stackrel{s_i}{\rightarrowtail} X, \Psi_i)$ is a closed table,
		\item $(S_i, \hat{\gamma})$ is a conjecture for this table,
		\item $s_i \circ \hat{s}_0 = s_0$, and
		\item $\Psi_{i+1}$ is a counterexample for $(S_i,\hat{\gamma},\hat{s}_0)$. 
	\end{itemize}
	We will show that such a sequence is necessarily finite.	
	
	By the last point and Lemma~\ref{lm:subform-coalg}, for each $i$, 
	there exists a map $q_{i+1,i}$ making the diagram on the left-hand side commute:
	$$
	\begin{tikzcd}
		& X \ar[dl,"{e_{\Psi_i}}"'] \ar[dr,"{e_{\Psi_{i+1}}}"] 
		& \\
		\xvert{X}{\Psi_i}
		& & \xvert{X}{\Psi_{i+1}} \ar[ll,"q_{i+1,i}"']
	\end{tikzcd}
	\qquad
	\begin{tikzcd}
		& X \ar[dl,"{e_{\Psi_i}}"'] \ar[dr,"{e_{\Phi}}"] 
		& \\
		\xvert{X}{\Psi_i}
		& & \xvert{X}{\Phi} \ar[ll,"q_{i}"']
	\end{tikzcd}
	$$
	Moreover, again by Lemma~\ref{lm:subform-coalg}, for each $i$,
	there is a map $q_i$ making the diagram on the right-hand side above commute. 
	For each $i$, we have
	$$
	q_{i+1,i} \circ q_{i+1} \circ e_{\Phi} 
		= q_{i+1,i} \circ e_{i+1} 
		= e_{\Psi_i}
		= q_i \circ e_{\Phi}
	$$
	and since $e_{\Phi}$ is epic, we obtain $q_{i+1, i} \circ q_{i+1} = q_i$. Hence, we get 
	the following sequence of quotients:
$$
\begin{tikzcd}
	& & \xvert{X}{\Phi} \ar[dll, "q_0"', twoheadrightarrow] 
	 \ar[dl, "q_1", twoheadrightarrow] 
	  \ar[d, "q_2", twoheadrightarrow] 
	   \ar[dr, "{\ldots}", twoheadrightarrow] 
	\\
	\xvert{X}{\Psi_0}
	& \xvert{X}{\Psi_1} \ar[l, "q_{1,0}", twoheadrightarrow]
	& \xvert{X}{\Psi_2} \ar[l, "q_{2,1}", twoheadrightarrow]
	& \ldots \ar[l, "q_{3,2}", twoheadrightarrow]
\end{tikzcd}
$$
It follows from Lemma~\ref{lm:counterexample-iso} and the previous assumptions that none of the 
quotients $q_{i+1,i}$ can be an iso.
But since for each $i$,  $\xvert{X}{\Psi_i}$ is a quotient of $\xvert{X}{\Phi}$,
and the latter has only finitely many quotients, the sequence of
counterexamples must be finite. 
\qed
\end{proof}

\end{appendix}

\end{document}